\documentclass[runningheads,a4paper]{llncs}

\usepackage{amssymb}
\usepackage{graphicx}

\usepackage{epstopdf}
\usepackage{subfigure}
\usepackage{amsmath}
\usepackage{algorithm}
\usepackage{algorithmic}
\usepackage[algo2e]{algorithm2e}

\usepackage{physics}

\usepackage{url}
\urldef{\mailsa}\path|bo.han@student.uts.edu.au|
\urldef{\mailsb}\path|{ivor.tsang,ling.chen}@uts.edu.au,xkxiao@nus.edu.sg|
\urldef{\mailsc}\path|sffung@cityu.edu.hk,celina.yu@gbca.edu.au|

\begin{document}

\mainmatter  

\title{Privacy-preserving Stochastic Gradual Learning}

\titlerunning{Privacy-preserving Stochastic Gradual Learning}

%
%
\author{Bo Han$^{1}$%
\and Ivor W. Tsang$^{1}$ \and Xiaokui Xiao$^2$ \and Ling Chen$^1$,\\ Sai-Fu Fung$^3$ \and Celina P. Yu$^4$}
\authorrunning{Bo Han et al.}

\institute{$^1$Center for Artificial Intelligence, University of Technology Sydney, Australia\\
$^2$Department of Computer Science, National University of Singapore, Singapore\\
$^3$Department of Applied Social Sciences, City University of Hong Kong, Hong Kong\\
$^4$Global Business College of Australia, Australia\\
\mailsa\\
\mailsb\\
\mailsc\\
}

%
%
\maketitle

\begin{abstract}
It is challenging for stochastic optimizations to handle large-scale sensitive data safely. Recently, Duchi et al. proposed private sampling strategy to solve privacy leakage in stochastic optimizations. However, this strategy leads to robustness degeneration, since this strategy is equal to the noise injection on each gradient, which adversely affects updates of the primal variable. To address this challenge, we introduce a robust stochastic optimization under the framework of local privacy, which is called Privacy-pREserving StochasTIc Gradual lEarning (PRESTIGE). PRESTIGE bridges private updates of the primal variable (by private sampling) with the gradual curriculum learning (CL). Specifically, the noise injection leads to the issue of label noise, but the robust learning process of CL can combat with label noise. Thus, PRESTIGE yields ``private but robust'' updates of the primal variable on the private curriculum, namely an reordered label sequence provided by CL. In theory, we reveal the convergence rate and maximum complexity of PRESTIGE. Empirical results on six datasets show that, PRESTIGE achieves a good tradeoff between privacy preservation and robustness over baselines.
\end{abstract}

\section{Introduction}
{L}{earning} from large-scale sensitive data stems from 2004~\cite{quere2004mining}, and still keeps vibrant~\cite{hao2011privacy,wu2014data}. However, the direct use of learning algorithms on such data will lead to the issues of ``computational burden''~\cite{bottou2016optimization} and ``privacy leakage'' \cite{xiong2005privacy}. Unfortunately, large-scale sensitive data are ubiquitous in the real world, such as electronic health record~\cite{gelfand2012privacy}, mobile app information~\cite{barbara1999mobile,zhu2014mobile}, and genome-wide association database~\cite{jiang2004cluster,jiang2005interactive,wang2016learning}.

Recently, researchers leverage stochastic optimization to handle large-scale data, because of its low computational cost~\cite{bottou2016optimization}. First, it does not require to compute the full gradient in each iteration. This merit reduces time costs greatly. Second, in each iteration, it processes either a single point~\cite{cotter2011better} or a tiny batch of points~\cite{mitliagkas2013memory}. However, on large-scale sensitive data, the direct use of stochastic optimization is obviously unsafe, which may lead to the issue of ``privacy leakage''.

Specifically, large-scale sensitive data itself (i.e., electronic health record) carries too much private information (i.e., disease history or diet habit) related to data providers. Even if we leverage the ``safe harbor'' method~\cite{benitez2010evaluating,malin2011never} to de-identify~\footnote{The de-identification involves removing many identifiers, such as names, dates, genders, drive license and social security numbers.} these sensitive data before using it, the re-identification is still possible~\cite{gelfand2012privacy}. Therefore, the output of stochastic optimization will carry implicit information delivered from large-scale sensitive data. Based on these information, hackers can perform reverse engineering, namely reasoning from output to input, to re-identify data providers and expose their privacy~\cite{fung2007anonymizing,yu2014scalable}.

To handle large-scale sensitive data, it is reasonable to propose the privacy-preserving stochastic optimizations \cite{song2013stochastic} under the global privacy \cite{dwork2006calibrating,xiao2011differential,zhu2017differentially}. Motivated by the fact that the local privacy is more stringent than the global privacy~\cite{warner1965randomized}, Duchi et al. designed a private sampling strategy to preserve the privacy of stochastic optimizations under local privacy \cite{duchi2013local-focs}. However, the issue of ``robustness degeneration'' arises subsequently, since the technique of private sampling is essentially equal to the noise injection on each gradient, which adversely affects updates of the primal variable. This issue is very common for privacy-preserving algorithms \cite{kantarcioglu2004privacy,liu2006random,dwork2009differential,chaudhuri2011differentially,wang2016learning}.

To address this challenge, we introduce a robust stochastic optimization  under the framework of local privacy \cite{aggarwal2008privacy}, which is called Privacy-pREserving StochasTIc Gradual lEarning (PRESTIGE). PRESTIGE naturally bridges private updates of the primal variable (by private sampling) with the gradual learning process of curriculum learning (CL) \cite{bengio2009curriculum}. Our inspiration springs from the learning process of CL, namely learning from ``easy'' tasks to ``complex'' tasks, which is often used for training robust models \cite{bengio2009curriculum,kumar2010self}. Specifically, the noise injection leads to the issue of label noise, but the robust learning process of CL can combat with label noise \cite{chen2015webly,han2018progressive}. Therefore, PRESTIGE achieves ``private but robust'' updates of the primal variable on the private curriculum, which is an reordered label sequence provided by CL from beneficial labels to adverse labels \footnote{Beneficial label is sufficiently reliable for the update of the primal variable correctly. Meanwhile, adverse label is unreliable or even noisy for the update of the primal variable correctly.}. To sum up, in the first epoch, PRESTIGE ensures the update of the primal variable on beneficial labels, which creates a robust model from the outset. In subsequent epochs, updates of the primal variable occur on adverse labels gradually until convergence. \textbf{Our contributions} are summarized as follows.
\begin{enumerate}
\item We introduce a ``private but robust'' stochastic optimization called PRESTIGE. This is the first work to solve robustness degeneration in privacy-preserving stochastic optimizations by curriculum learning.
\item We define the private curriculum to realize PRESIGE, and reveal the convergence rate and the maximum complexity (KWIK bound \cite{li2008knows}) of PRESTIGE.
\item We conduct comprehensive experiments on UCI and real-world sensitive datasets. Empirical results show that PRESTIGE achieves a good tradeoff between privacy preservation and robustness over other baselines.
\end{enumerate}

\section{Related Work}\label{RelatedWork}
%
Our study deals with privacy-preserving online/stochastic optimizations for large-scale sensitive data. For example, Jain et al. proposed the differentially-private online learning \cite{jain2012differentially}. Song et al. proposed the private stochastic gradient descent (SGD) by the gradient perturbation, and improved this work with the mini-batch trick \cite{song2013stochastic}. Wu et al. bridged differentially-private SGD with a practical RDBMS system \cite{wu2017bolt}. Nevertheless, all their works are based on global differential privacy \cite{hu2015differential}. Since local differential privacy is more stringent than global differential privacy, Duchi et al. proposed a private SGD under local privacy \cite{duchi2013local-focs}. However, their method does not consider ``robustness degeneration''. Our PRESTIGE aims to address ``privacy leakage'' and ``robustness degeneration'' simultaneously under local privacy.

Our work is also related to curriculum learning (CL) and self-paced learning. Bengio et al. provided a learning paradigm called curriculum learning \cite{bengio2009curriculum}, and Kumar et al. presented the similar learning regime named self-paced learning \cite{kumar2010self}. The idea shared by these two studies is to learn easier tasks first, and gradually learn more difficult tasks to result in a robust model. However, CL has never been applied into any privacy-preserving algorithms due to two challenges. First, CL is a high-level idea without the specific formalization, it is still unknown how to realize the private CL algorithms under privacy constraints. Moreover, existing CL algorithms do not have any formal termination criteria, which needs more iterations to converge. PRESTIGE can be viewed as the first work to realize the private CL under local privacy.

Therefore, in the following paper, we try to address three important questions: 1) why curriculum learning can solve the issue of robustness degeneration? 2) how to design the private curriculum to realize PRESTIGE? 3) how to provide theoretical guarantees for PRESTIGE?

\section{PRESTIGE}\label{SelfPacedSL}
In this section, we begin with preliminary notations and definitions. Then, we briefly present why ``robustness degeneration'' will happen, and how to overcome this nontrivial issue by ``private curriculum''. Lastly, we provide the theoretical analysis.

\subsection{Preliminary Notations and Definitions}\label{Preliminary}
Throughout the paper, let $\mathcal{D} = \{\mathbf{x}_i, y_i\}_{i=1}^n $ be the training data, where $\mathbf{x}_i \in \mathbb{R}^d$ denotes the $i$th instance and $y_i \in \{-1,+1\}$ denotes its binary label. A typical classification model is represented as:
\begin{equation}\label{obj-classification}
\min_{\mathbf{w}} F(\mathbf{w})= \min_{\mathbf{w}}\frac{1}{n}\sum_{i=1}^{n} f_i(\mathbf{w}),
\end{equation}
where $\mathbf{w} \in \mathbb{B}_d(R)$ is the primal variable. $\mathbb{B}_d(R)$ denotes the $d$ dimensional Euclidean ball of radius $R$. Specifically, $f_i(\mathbf{w}) = \rho_\lambda(\mathbf{w}) + r(\mathbf{w};\{\mathbf{x}_i,y_i\})$ where $\lambda$ is the regularization parameter, $\rho_\lambda(\mathbf{w})$ is the regularizer, and $r(\mathbf{w};\{\mathbf{x}_i,y_i\})$ is a loss function. Note that, robustness can be evaluated by the accuracy of classification model.

Intuitively, the $\epsilon$-differentially globally private algorithm preserves the privacy on the level of the whole dataset (Definition~\ref{G-Privacy}). However, the $\epsilon$-differentially locally private algorithm preserves the privacy on the level of each instance (Definition~\ref{L-Privacy}). Therefore, local privacy is more stringent than the global privacy~\cite{warner1965randomized}.

\begin{definition}\label{G-Privacy}
(Global Privacy, Definition 2 in~\cite{dwork2006}). For a privacy parameter $\epsilon \geq 0$, a algorithm $\mathcal{M}$ is $\epsilon$-differentially globally private, if

\begin{equation}
\Pr(\mathcal{M}(\mathbf{X}) \in S) \leq \exp(\epsilon)\Pr(\mathcal{M}(\mathbf{X}') \in S),
\end{equation}
for all measurable subsets $S$ of the range of $\mathcal{M}$, and for all datasets $\mathbf{X}$, $\mathbf{X}'$ differing by a single entry (their Hamming distance equals to $1$).
\end{definition}

\begin{definition}\label{L-Privacy}
(Local Privacy, Definition 1 in~\cite{duchi2018minimax}) For a privacy parameter $\epsilon \geq 0$, a algorithm $\mathcal{M}$ is $\epsilon$-differentially locally private, if

\begin{equation}
\sup_{S \in \sigma(\mathcal{Z})}\sup_{x,x' \in \mathcal{X}}\frac{\Pr(\mathbf{z}_i \in S|\mathbf{x}_i = x)}{\Pr(\mathbf{z}_i \in S|\mathbf{x}_i = x')}\leq \exp(\epsilon),
\end{equation}
where $\mathbf{z}_i = \mathcal{M}(\mathbf{x}_i)$ for $\forall i$. Namely, the output $\{\mathbf{z}_i\}_{i=1}^{n}$ is an $\epsilon$-differentially locally private view of input $\{\mathbf{x}_i\}_{i=1}^{n}$. $\sigma(\mathcal{Z})$ denotes an appropriate $\sigma$-field on $\mathcal{Z}$.
\end{definition}

\subsection{Challenge}\label{Privacy-Preservation}
\begin{figure}
\begin{center}
\centerline{\includegraphics[width=0.55\textwidth]{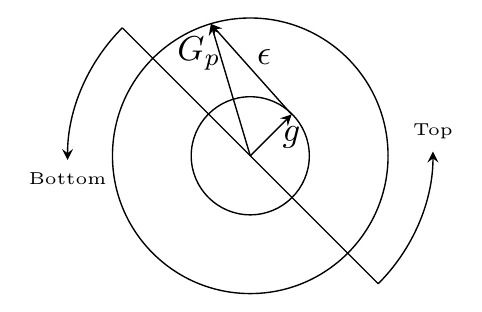}}
\caption{Private sampling perturbs and rescales gradient $g$ into $G_p$, which equals to noise $\epsilon$ injection.}
\label{Privacy-Figure}
\end{center}
\end{figure}

\begin{proposition}\label{LocallyPrivateSampling}
(Private Sampling for Stochastic Gradient) Given a stochastic gradient vector $g$ with $\lVert g\rVert_2 \leq L$, where $L$ is Lipschitz parameter, rescale $g$ into $\tilde{g} = L\frac{g}{\lVert g\rVert_2}$ with probability $\frac{1}{2} + \frac{\lVert g\rVert_2}{2L}$ and $\tilde{g} = -L\frac{g}{\lVert g\rVert_2}$ with probability $\frac{1}{2} - \frac{\lVert g\rVert_2}{2L}$. After sampling $Q \sim$ \text{Bernoulli}($\pi_{\varepsilon_s} = \frac{e^{\varepsilon_s}}{e^{\varepsilon_s} + 1}$), we sample $G_p$ (privatized stochastic gradient vector) as follows:
\begin{equation}\label{private-gradient}
G_p \sim \left\{
\begin{array}{cc}
\text{Unif}(\langle g_p,\tilde{g}\rangle > 0), &Q = 1\\
\text{Unif}(\langle g_p,\tilde{g}\rangle \leq 0), &Q = 0\\
\end{array},
\right.
\end{equation}
where $g_p \in \mathbb{R}^d$ and $\lVert g_p\rVert_2 = B$.
\end{proposition}

\begin{proof}
Here, we prove whether private sampling belongs to the non-interactive strategy of $\varepsilon_s$-local differential privacy. Note that, we leverage Figure \ref{Privacy-Figure} to assist our proof. By the first principle, we should start from the adaptive definition of $\varepsilon_s$-local differential privacy, and we employ the non-interactive strategy:

\begin{equation}
\sup_{R}\sup_{r,r'}\frac{\Pr(G_P\in R|g\in r)}{\Pr(G_P\in R|g\in r')}\leq e^{\varepsilon_s}.
\end{equation}
Note that, $R, r, r'$ represents the ``Top'' or ``Bottom'' area (``Top'' or ``Bottom'' is abbreviated as ``$\mathbb{T}$'' or ``$\mathbb{B}$''), where $G_p$ or $g$ locates. Therefore, this proof can be conducted from four cases as follows: 1) $\frac{\Pr(G_P\in \mathbb{T}|g\in \mathbb{T})}{\Pr(G_P\in \mathbb{T}|g\in \mathbb{B})}$; 2) $\frac{\Pr(G_P\in \mathbb{B}|g\in \mathbb{B})}{\Pr(G_P\in \mathbb{B}|g\in \mathbb{T})}$; 3) $\frac{\Pr(G_P\in \mathbb{T}|g\in \mathbb{B})}{\Pr(G_P\in \mathbb{T}|g\in \mathbb{T})}$; 4) $\frac{\Pr(G_P\in \mathbb{B}|g\in \mathbb{T})}{\Pr(G_P\in \mathbb{B}|g\in \mathbb{B})}$. Assume that $\frac{\lVert g\rVert_2}{2L} = k$, due $\lVert g\rVert_2 \leq L$, we have $k \leq 0.5$.

If $g\in \mathbb{T}$, then $\Pr(\tilde{g}\in \mathbb{T}|g\in \mathbb{T}) = 0.5 + \frac{\lVert g\rVert_2}{2L}$, meanwhile, if $\tilde{g}\in \mathbb{T}$, then $\Pr(G_p\in \mathbb{T}|\tilde{g}\in \mathbb{T}) = \frac{e^{\varepsilon_s}}{e^{\varepsilon_s}+1}$. However, if $g\in \mathbb{T}$, then $\Pr(\tilde{g}\in \mathbb{B}|g\in \mathbb{T}) = 0.5 - \frac{\lVert g\rVert_2}{2L}$, meanwhile, if $\tilde{g}\in \mathbb{B}$, then $\Pr(G_p\in \mathbb{T}|\tilde{g}\in \mathbb{B}) = \frac{1}{e^{\varepsilon_s}+1}$. According to the chain rule, we have $\Pr(G_p\in \mathbb{T}|g\in \mathbb{T}) = (0.5 + k)\frac{e^{\varepsilon_s}}{e^{\varepsilon_s}+1} + (0.5 - k)\frac{1}{e^{\varepsilon_s}+1}$.

If $g\in \mathbb{B}$, then $\Pr(\tilde{g}\in \mathbb{T}|g\in \mathbb{B}) = 0.5 - \frac{\lVert g\rVert_2}{2L}$, meanwhile, if $\tilde{g}\in \mathbb{T}$, then $\Pr(G_p\in \mathbb{T}|\tilde{g}\in \mathbb{T}) = \frac{e^{\varepsilon_s}}{e^{\varepsilon_s}+1}$. However, if $g\in \mathbb{B}$, then $\Pr(\tilde{g}\in \mathbb{B}|g\in \mathbb{B}) = 0.5 + \frac{\lVert g\rVert_2}{2L}$, meanwhile, if $\tilde{g}\in \mathbb{B}$, then $\Pr(G_p\in \mathbb{T}|\tilde{g}\in \mathbb{B}) = \frac{1}{e^{\varepsilon_s}+1}$. According to the chain rule, we have $\Pr(G_p\in \mathbb{T}|g\in \mathbb{B}) = (0.5 - k)\frac{e^{\varepsilon_s}}{e^{\varepsilon_s}+1} + (0.5 + k)\frac{1}{e^{\varepsilon_s}+1}$.

Therefore, we have:
\begin{equation}\label{case1}
\begin{split}
&\quad\sup\sup\frac{\Pr(G_p\in \mathbb{T}|g\in \mathbb{T})}{\Pr(G_p\in \mathbb{T}|g\in \mathbb{B})}\\
&= \frac{0.5 + \frac{e^{\varepsilon_s}-1}{e^{\varepsilon_s}+1}k}{0.5 + \frac{1-e^{\varepsilon_s}}{e^{\varepsilon_s}+1}k} \leq \frac{0.5 + \frac{e^{\varepsilon_s}-1}{e^{\varepsilon_s}+1}0.5}{0.5 + \frac{1-e^{\varepsilon_s}}{e^{\varepsilon_s}+1}0.5} = \frac{1 + \frac{e^{\varepsilon_s}-1}{e^{\varepsilon_s}+1}}{1 + \frac{1-e^{\varepsilon_s}}{e^{\varepsilon_s}+1}} = e^{\varepsilon_s}.
\end{split}
\end{equation}
Therefore, the first case has been proved. By symmetry, we prove the second case in the following.

If $g\in \mathbb{B}$, then $\Pr(\tilde{g}\in \mathbb{T}|g\in \mathbb{B}) = 0.5 - \frac{\lVert g\rVert_2}{2L}$, meanwhile, if $\tilde{g}\in \mathbb{T}$, then $\Pr(G_p\in \mathbb{B}|\tilde{g}\in \mathbb{T}) = \frac{1}{e^{\varepsilon_s}+1}$. However, if $g\in \mathbb{B}$, then $\Pr(\tilde{g}\in \mathbb{B}|g\in \mathbb{B}) = 0.5 + \frac{\lVert g\rVert_2}{2L}$, meanwhile, if $\tilde{g}\in \mathbb{B}$, then $\Pr(G_p\in \mathbb{B}|\tilde{g}\in \mathbb{B}) = \frac{e^{\varepsilon_s}}{e^{\varepsilon_s}+1}$. According to the chain rule, we have $\Pr(G_p\in \mathbb{B}|g\in \mathbb{B}) = (0.5 - k)\frac{1}{e^{\varepsilon_s}+1} + (0.5 + k)\frac{e^{\varepsilon_s}}{e^{\varepsilon_s}+1}$.

If $g\in \mathbb{T}$, then $\Pr(\tilde{g}\in \mathbb{T}|g\in \mathbb{T}) = 0.5 + \frac{\lVert g\rVert_2}{2L}$, meanwhile, if $\tilde{g}\in \mathbb{T}$, then $\Pr(G_p\in \mathbb{B}|\tilde{g}\in \mathbb{T}) = \frac{1}{e^{\varepsilon_s}+1}$. However, if $g\in \mathbb{T}$, then $\Pr(\tilde{g}\in \mathbb{B}|g\in \mathbb{T}) = 0.5 - \frac{\lVert g\rVert_2}{2L}$, meanwhile, if $\tilde{g}\in \mathbb{B}$, then $\Pr(G_p\in \mathbb{B}|\tilde{g}\in \mathbb{B}) = \frac{e^{\varepsilon_s}}{e^{\varepsilon_s}+1}$. According to the chain rule, we have $\Pr(G_p\in \mathbb{B}|g\in \mathbb{T}) = (0.5 + k)\frac{1}{e^{\varepsilon_s}+1} + (0.5 - k)\frac{e^{\varepsilon_s}}{e^{\varepsilon_s}+1}$.

Therefore, we have:
\begin{equation}\label{case2}
\begin{split}
&\quad\sup\sup\frac{\Pr(G_p\in \mathbb{B}|g\in \mathbb{B})}{\Pr(G_p\in \mathbb{B}|g\in \mathbb{T})}\\
&= \frac{0.5 + \frac{e^{\varepsilon_s}-1}{e^{\varepsilon_s}+1}k}{0.5 + \frac{1-e^{\varepsilon_s}}{e^{\varepsilon_s}+1}k} \leq \frac{0.5 + \frac{e^{\varepsilon_s}-1}{e^{\varepsilon_s}+1}0.5}{0.5 + \frac{1-e^{\varepsilon_s}}{e^{\varepsilon_s}+1}0.5} = \frac{1 + \frac{e^{\varepsilon_s}-1}{e^{\varepsilon_s}+1}}{1 + \frac{1-e^{\varepsilon_s}}{e^{\varepsilon_s}+1}} = e^{\varepsilon_s}.
\end{split}
\end{equation}

Thus, the second case has been proved. Furthermore, according to the above analysis and Eq. \eqref{case1} and \eqref{case2}, we can easily prove the case 3 and case 4, namely:
\begin{equation}
\quad\sup\sup\frac{\Pr(G_p\in \mathbb{T}|g\in \mathbb{B})}{\Pr(G_p\in \mathbb{T}|g\in \mathbb{T})} = \frac{1 + \frac{1-e^{\varepsilon_s}}{e^{\varepsilon_s}+1}}{1 + \frac{e^{\varepsilon_s}-1}{e^{\varepsilon_s}+1}} = \frac{1}{e^{\varepsilon_s}} \leq e^{\varepsilon_s}.
\end{equation}

\begin{equation}
\quad\sup\sup\frac{\Pr(G_p\in \mathbb{B}|g\in \mathbb{T})}{\Pr(G_p\in \mathbb{B}|g\in \mathbb{B})} = \frac{1 + \frac{1-e^{\varepsilon_s}}{e^{\varepsilon_s}+1}}{1 + \frac{e^{\varepsilon_s}-1}{e^{\varepsilon_s}+1}} = \frac{1}{e^{\varepsilon_s}} \leq e^{\varepsilon_s}.
\end{equation}

Then the proof completes. To sum up, we prove that private sampling for stochastic gradient belongs to the non-interactive strategy of $\varepsilon_s$-local differential privacy.
\end{proof}

Recently, Duchi et al. invented a novel strategy called private sampling \cite{duchi2013local-focs}, which imposes stringent local privacy for any $d$-dimensional mean-estimation problem \cite{duchi2013local-arXiv}. Therefore, this strategy can be easily adopted into stochastic gradient to preserve the privacy of stochastic optimizations (Proposition~\ref{LocallyPrivateSampling}). However, this strategy inevitably leads to ``robustness degeneration'', since this strategy is equal to the noise injection on each gradient (Figure~\ref{Privacy-Figure}), which adversely affects updates of the primal variable $\mathbf{w}$. Moreover, its unrestricted noise injection may lead to the slow convergence or even the divergence \cite{duchi2013local-focs}. Thus, these issues motivate us to explore a mechanism to ensure ``private but robust'' updates of the primal variable $\mathbf{w}$ with the convergence.
\subsection{Solution}\label{Sec:PRESTIGE}
\begin{table*}
\caption{Comparison of different approaches. DJW denotes Duchi-Jordan-Wainwright's approach using private sampling \cite{duchi2013local-focs}.}
\label{Comparison}
\begin{center}
\scalebox{1}{
\begin{tabular}{c|c|c|c}
\hline
Methods & Iterative Training & Local Privacy & Learning Process \\ \hline
SGD & Gradient-based & No &Random \\ \hline
DJW & Gradient-based & Yes &Random \\ \hline
CL & Heuristic & No & ``Easy'' to ``Complex'' \\ \hline
PRESTIGE & Gradient-based & Yes & ``Beneficial" to ``Adverse'' \\ \hline
\end{tabular}}
\end{center}
\end{table*}

\subsubsection{Main Idea}

Our idea is motivated by curriculum learning (CL), which learns easier tasks first, and learns more difficult tasks gradually to ensure a robust model. This strategy resembles training an infant through to adulthood. Namely, knowledge is ordered to allow for gradual learning. Infants are provided with easy knowledge first. As they grow and are able to handle more complex concepts, more difficult knowledge is provided. Based on the above inspiration, we introduce a ``private but robust'' stochastic optimization called ``Privacy-pREserving StochasTIc Gradual lEarning'' (PRESTIGE) for large-scale sensitive data. PRESTIGE incorporates private updates of the primal variable $\mathbf{w}$ (by private sampling) with the gradual learning of CL.

Specifically, the noise injection leads to the issue of label noise \cite{natarajan2013learning}. Meanwhile, CL provides the ordered learning to learn from easy tasks first then to hard tasks until convergence, which can combat with label noise \cite{chen2015webly}. Therefore, through the robust learning paradigm of CL, PRESTIGE aims to yield ``private but robust'' updates of the primal variable $\mathbf{w}$ on the curriculum (Definition~\ref{Curriculum}), which is an reordered label sequence from ``beneficial'' labels to ``adverse'' labels (Figure~\ref{curriculum-learning}). Table~\ref{Comparison} shows the key comparison of different learning approaches, and PRESTIGE integrates all benefits simultaneously. Noted that, CL is only a high-level idea without any specific formalization. We illustrate our PRESTIGE using a label noise setting as follows.

\begin{figure}[!tp]
\begin{center}
\centerline{\includegraphics[width=0.85\textwidth]{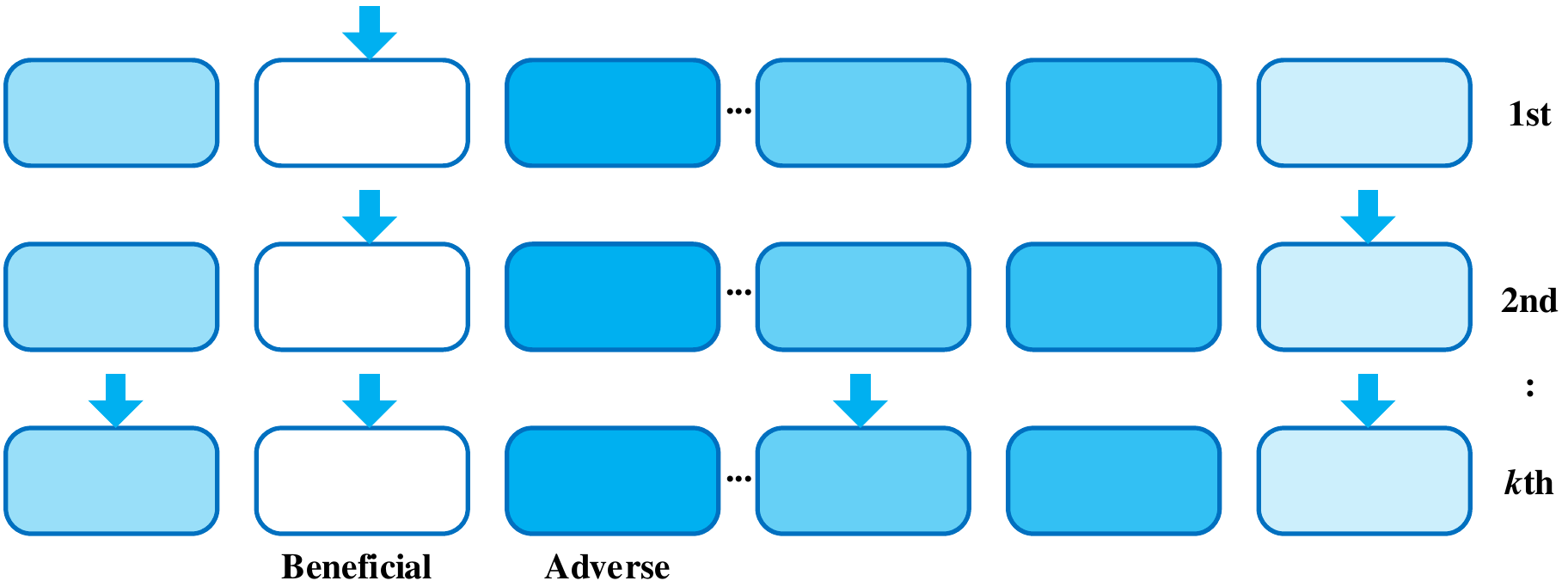}}
\caption{PRESTIGE yields robust updates of the primal variable {\bf w} on an reordered label sequence, namely from beneficial labels (light block) to adverse labels (dark block). The darker block denotes the curriculum with more adverse labels, and vice versa. In each iteration, arrows specify the feasible areas for updates of the primal variable {\bf w}.}
\label{curriculum-learning}
\end{center}
\end{figure}

\subsubsection{Illustrated Example} Given that private sampling adversely affects the update of current $\mathbf{w}$ (hyperplane). In the top panel of Figure~\ref{RobustLosses}, instance $\mathbf{x}_{A}$ (i.e., data point ``A'') should be with label $y_{A} = +1$, which corresponds to its predicted label value (+1) (predicted label value = $\left\{
\begin{array}{cc}
+1  &\mathcal{C}_{\mathbf{w}}(\mathbf{x}) \geq 0\\
-1  &\mathcal{C}_{\mathbf{w}}(\mathbf{x}) < 0\\
\end{array}
\right.
$, where $\mathcal{C}_\mathbf{w}$ denotes the current classifier). Therefore, the label $y_{A}$ of instance $\mathbf{x}_{A}$ can be regarded as a ``beneficial'' label. Namely, beneficial label is sufficiently reliable for the update of the primal variable correctly. Conversely, instances $\mathbf{x}_{B}$, $\mathbf{x}_{C}$ and $\mathbf{x}_{D}$ (i.e., data points ``B'', ``C'' and ``D'') should be with labels \mbox{$y_{B} = -1, y_{C} = -1, y_{D} = +1$}. These labels are the opposite of their predicted label value. Therefore, the label $y_{B}$ of instance $\mathbf{x}_{B}$, the label $y_{C}$ of instance $\mathbf{x}_{C}$ and the label $y_{D}$ of instance $\mathbf{x}_{D}$ can be regarded as ``adverse'' labels. Namely, adverse label is unreliable or even noisy for the update of the primal variable correctly. Among them, ``B'' is farther than ``C'' and ``D'' to the hyperplane, which means that ``B'' is more unreliable. More importantly, these labels would negatively affect the update of the primal variable. To remedy this negative effect, PRESTIGE attempts to leverage the robust learning regime of CL. First, we define the curriculum of PRESTIGE.

\begin{figure}[!tp]
\begin{center}
\centerline{\includegraphics[width=0.55\textwidth]{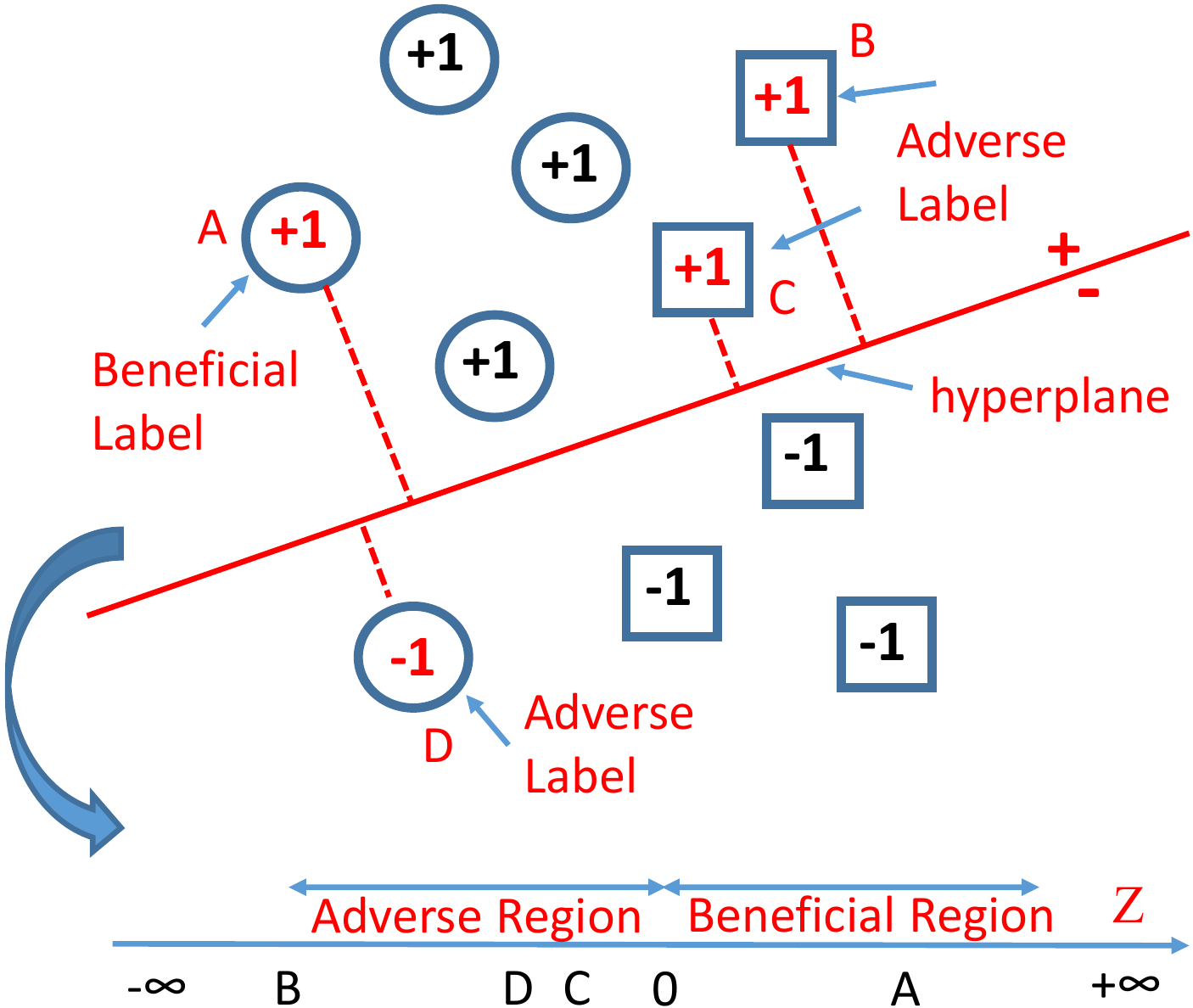}}
\caption{\textbf{Top Panel:} Circles denote \emph{real positive} instances such as ``A'' and ``D''. Squares represent \emph{real negative} instances such as ``B'' and ``C''; however, both ``B'' and ``C'' are \emph{erroneously} viewed as positive class, and ``D'' is \emph{erroneously} viewed as negative class. This creates three adverse labels. According to their distance to the hyperplane, the label of ``A'' is beneficial; while the labels of ``B'', ``C'' and ``D'' are adverse. \textbf{Bottom Panel:} ``A'' is located in the beneficial region defined by $z \geq 0$; while ``B'', ``C'' and ``D'' are located in the adverse region defined by $z \leq 0$.}
\label{RobustLosses}
\end{center}
\end{figure}

\begin{definition}\label{Curriculum}
(Curriculum). For training data $\mathcal{D} = \{\mathbf{x}_i, y_i\}_{i=1}^n $, if $\mathcal{C}_{\mathbf{w}}$ denotes the current classifier, then the curriculum $z$ in PRESTIGE can be calculated as the product of the label $y$ and the predicted label of an instance $\mathbf{x}$, namely $z = y\mathcal{C}_{\mathbf{w}}(\mathbf{x})$.
\end{definition}

\begin{remark}
In PRESTIGE, $z_i > z_j$ represents that $x_i$ is more reliable than $x_j$, which further means that $x_i$ should be learned earlier than $x_j$. In conclusion, PRESTIGE learns from ``beneficial'' labels to ``adverse'' labels.
\end{remark}

According to Definition \ref{Curriculum}, given beneficial data $\{\mathbf{x}_A, y_A\}$, adverse data $\{\mathbf{x}_B, y_B\}$, $\{\mathbf{x}_C, y_C\}$ and $\{\mathbf{x}_D, y_D\}$ in the top panel of \mbox{Figure~\ref{RobustLosses}}, we have the curriculum $z_A = y_{A}\mathcal{C}_{\mathbf{w}}(\mathbf{x}_{A}) > 0$, the curriculums $z_B = y_{B}\mathcal{C}_{\mathbf{w}}(\mathbf{x}_{B}) < z_D = y_{D}\mathcal{C}_{\mathbf{w}}(\mathbf{x}_{D}) < z_C = y_{C}\mathcal{C}_{\mathbf{w}}(\mathbf{x}_{C}) < 0$. Therefore, in the bottom panel of Figure~\ref{RobustLosses}, the curriculum $z = 0$ can be employed to separate the axis into the beneficial region ($z \geq 0$) and the adverse region ($z \leq 0$) respectively. Furthermore, the update region for PRESTIGE ($z \geq D_{th}$) is controlled by a dynamic threshold $D_{th}$. Note that, $D_{th}$ should be initialized larger than~$1$ in the axis and the curriculum $z \geq D_{th}$. Thus, the curriculum $z > 1$ includes beneficial labels by the max-margin principle~\cite{shalev2014understanding}. Through the robust learning process of CL, in the initial epoch, the update of the primal variable $\mathbf{w}$ is limited to ``beneficial" labels ($z \geq D_{th} > 1$) to establish a robust model. In the following epochs, we gradually reduce the dynamic threshold $D_{th}$. Correspondingly, updates occur from ``beneficial'' to ``adverse'' labels incrementally until convergence.

\subsubsection{Design of Private Curriculum} The former example explains how PRESTIGE yields robust updates of the primal variable {\bf w} on the curriculum provided by CL. However, this curriculum may break the constraint of local privacy, which motivates us to propose the private curriculum (Definition~\ref{P-Curriculum}) for PRESTIGE. Reasons are explained in Remark of Definition~\ref{P-Curriculum}. Note that, Proposition~\ref{randomizedresponse} serves for Definition~\ref{P-Curriculum}.
\begin{proposition}\label{randomizedresponse} Given the design matrix $P_{c}$ of Warner's model~\cite{warner1965randomized}. If randomized response is leveraged to preserve $\varepsilon_r$-differentially local privacy, the design matrix $P_{c}$ should be:
\begin{equation}\label{Pc}
P_c =
\begin{pmatrix}
\frac{e^{\varepsilon_r}}{e^{\varepsilon_r}+1}&\frac{1}{e^{\varepsilon_r}+1}\\
\frac{1}{e^{\varepsilon_r}+1}&\frac{e^{\varepsilon_r}}{e^{\varepsilon_r}+1}\\
\end{pmatrix}.
\end{equation}
\end{proposition}

\begin{remark}
We acquire noise rates ($\rho_+,\rho_-$) from the design matrix $P_c$ directly, and $\rho_+ = \rho_- = \frac{1}{e^{\varepsilon_r}+1}$, where $\rho_+ = \Pr(\tilde{y} = -1 | y = +1)$ and $\rho_- = \Pr(\tilde{y} = +1 | y = -1)$. Namely, we perturb the original $y$ to form $\tilde{y}$ by noise rates ($\rho_+,\rho_-$).
\end{remark}

\begin{proof}
According to the definition of design matrix $P_c$, $p_{+1,-1} = \rho_+ = \Pr(\tilde{y} = -1 | y = +1)$, $p_{-1,+1} = \rho_- = \Pr(\tilde{y} = +1 | y = -1)$, $p_{+1,+1} = \Pr(\tilde{y} = +1 | y = +1)$ and $p_{-1,-1} = \Pr(\tilde{y} = -1 | y = -1)$. Assume that $\frac{p_{+1,+1}}{p_{+1,-1}} = a$ and $\frac{p_{-1,-1}}{p_{-1,+1}} = b$. To meet $\varepsilon_r$-differentially local privacy, we have $a \leq e^{\varepsilon_r}$ and
$b \leq e^{\varepsilon_r}$. Here, we assume that randomized response still prefers the true value, namely $p_{+1,+1},p_{-1,-1} > 0.5$. Therefore, $1 < a \leq e^{\varepsilon_r}$ and
$1 < b \leq e^{\varepsilon_r}$. Since randomized response prefers the true value, we also hope to maximize $p_{+1,+1} + p_{-1,-1}$, which equally maximizes $\frac{2a}{a+1}$ by Warner's model ($a = b$). Due to $\frac{\partial(\frac{2a}{a+1})}{\partial a} = \frac{2}{(a+1)^2} > 0$, $\frac{2a}{a+1}$ is monotonically increasing in the feasible area $(1,e^{\varepsilon_r}]$. Therefore, we set $a = e^{\varepsilon_r}$ to achieve its maximum value, and the design matrix $P_c$ is Eq. \eqref{Pc}.
\end{proof}

\begin{definition}\label{P-Curriculum}
(Private Curriculum). For training data $\mathcal{D} = \{\mathbf{x}_i, y_i\}_{i=1}^n $, if $\mathcal{C}_{\mathbf{w}}$ denotes the current classifier, then the private curriculum $z$ in PRESTIGE can be calculated as the product of the perturbed label $\tilde{y}$ and the predicted label of an instance $\mathbf{x}$, where $\tilde{y}$ is perturbed from $y$ according to Proposition~\ref{randomizedresponse}. Namely, $z = \tilde{y}\mathcal{C}_{\mathbf{w}}(\mathbf{x})$, where $\tilde{y}$ is perturbed from $y$ by noise rates $\rho_+ = \rho_- = \frac{1}{e^{\varepsilon_r}+1}$.
\end{definition}

\begin{remark}\label{p-curriculum-remark}
In Algorithm~\ref{algorithm-PRESTIGE}, we use private curriculum (Definition~\ref{P-Curriculum}) instead of curriculum (Definition~\ref{Curriculum}). The reason is, without of label perturbation in Definition~\ref{P-Curriculum}, when the curriculum of sample $\{\mathbf{x}_{it}, y_{it}\}$ locates without the current update region, namely $z_{it}(\mathbf{w}^{(tmp)}) < D_{th}$, the algorithm itself will leak the information that this sample is adverse. This issue breaks the basic definition of local privacy.
\end{remark}

\subsubsection{Algorithm Realization}

\begin{algorithm*}
\KwIn{$\lambda \geq 0$, $b$, the max number of epochs $T_{max}$, the step size $\mu$, the loss function $r(\mathbf{w};\{\mathbf{x}_i,y_i\})$, the regularizer $\rho_\lambda(\mathbf{w}) = \frac{\lambda}{2}\lVert \mathbf{w} \rVert^2$, the training set $\mathcal{D} = \{\mathbf{x}_i, y_i\}_{i=1}^n$, two local privacy parameters $\varepsilon_r$, $\varepsilon_s \leq 1$ where $\varepsilon = \varepsilon_r + \varepsilon_s$, the Lipschitz parameter $L$ and the scalar bound $B \in \mathbb{R}_+$.}

{\bfseries Initialize:} $t = 0$, $\tilde{\mathbf{w}}^{(0)}$ randomly, the dynamic threshold \mbox{$D_{th} > 1$} by the max-margin principle (i.e., $D_{th} = 1.5$), the initial value of dynamic learning rate $\eta_0 = \frac{R}{B}.$

\For{$T = 1,2,\dotsc,T_{max}$}{
{\bfseries Assign:} $\mathbf{w}^{(tmp)} = \tilde{\mathbf{w}}^{(T-1)}$.

{\bfseries Shuffle:} $n$ training instances in $\mathcal{D}$.

\For{$k = 1,\dotsc,n$}{
   {\bfseries Sequentially pick:} $\{\mathbf{x}_{it},y_{it}\}$ from $\mathcal{D}$, $it \in \{1,...,n\}$.

   {\bfseries \underline{P-Curriculum:}} \mbox{$z_{it}(\mathbf{w}^{(tmp)}, \varepsilon_r) = (\langle \mathbf{w}^{(tmp)},\mathbf{x}_{it}\rangle + b)\tilde{y}_{it}$}, where $\tilde{y}_{it}$ from $y_{it}$ by noise rates $\rho_+ = \rho_- = \frac{1}{e^{\varepsilon_r}+1}$.

   {\bfseries If} \underline{$z_{it}(\mathbf{w}^{(tmp)},\varepsilon_r) \geq D_{th}$:}

\quad {\bfseries Update:} $t = t + 1$ and $\eta = \frac{\eta_0}{\lambda\sqrt{t}} = \frac{R}{\lambda B\sqrt{t}}$.

\quad {\bfseries Compute:} \mbox{$g = \lambda \mathbf{w}^{(tmp)} + \partial_{\mathbf{w}}r(\mathbf{w}^{(tmp)};\{\mathbf{x}_{it},\tilde{y}_{it}\})$}.

\quad {\bfseries Rescale:} $\tilde{g} = \left\{
\begin{array}{cc}
L\frac{g}{\lVert g\rVert_2},  &0.5 + \frac{\lVert g\rVert_2}{2L}\\
-L\frac{g}{\lVert g\rVert_2},  &0.5 - \frac{\lVert g\rVert_2}{2L}\\
\end{array}
\right.$.

\quad {\bfseries Sample:} $Q \sim \text{Bernoulli}(\frac{e^{\varepsilon_s}}{e^{\varepsilon_s}+1})$.

\quad {\bfseries Sample:} \mbox{$G_p \sim \left\{
\begin{array}{cc}
\text{Unif}(\langle g_p, \tilde{g}\rangle > 0), &Q = 1\\
\text{Unif}(\langle g_p, \tilde{g}\rangle \leq 0), &Q = 0\\
\end{array}
\right.$}, where $g_p \in \mathbb{R}^d$ and $\lVert g_p\rVert_2 = B$.

\quad {\bfseries \underline{P-Update:}} \mbox{$\mathbf{w}^{(new)} = \mathbf{w}^{(tmp)} - \eta G_p$}.

\quad {\bfseries Assign:} $\mathbf{w}^{(tmp)} = \mathbf{w}^{(new)}$.

}

 {\bfseries Assign:} $\tilde{\mathbf{w}}^{(T)} = \mathbf{w}^{(tmp)}$.

 {\bfseries Update:} $D_{th} = D_{th} - \mu\sqrt{T}$.
}
\KwOut{$\tilde{\mathbf{w}}^{(T_{max})}$.}
\caption{\textbf{PRESTIGE}: \textbf{P}rivacy-p\textbf{RE}serving \textbf{S}tochas\textbf{TI}c \textbf{G}radual l\textbf{E}arning \label{algorithm-PRESTIGE}}
\end{algorithm*}

We realize details of PRESTIGE in Algorithm~\ref{algorithm-PRESTIGE}, which consists of two key components: private curriculum (p-curriculum) calculation and judgement (line $8-9$ and $10$) and private update (p-update, line $17$). PRESTIGE preserves a composite privacy: p-curriculum privacy with parameter $\varepsilon_r$ (line $9$) and in-curriculum privacy with parameter $\varepsilon_s$ (line $14$) sequentially. Specifically, 1) PRESTIGE compute p-curriculum $z_{it}(\mathbf{w}^{(tmp)},\varepsilon_r)$ (Definition~\ref{P-Curriculum}) to preserve p-curriculum privacy, which is essentially based on randomized response~\cite{kairouz2016discrete}. 2) PRESTIGE conducts private sampling in line $11-16$ to preserve in-curriculum privacy, where we follow Duchi et al.~\cite{duchi2013local-focs}. 3) With the decrease of $D_{th}$ (line~$20$), PRESTIGE updates its primal variable from ``beneficial'' samples to ``adverse'' samples privately and robustly.

There is a point to be mentioned in Algorithm \ref{algorithm-PRESTIGE}. The private curriculum $z_{it}(\mathbf{w}^{(tmp)}, \varepsilon_r)$ is realized by the linear mapping function $\mathcal{C}_w$ in line~$8$. It means that we apply PRESTIGE to linear classification (i.e., SVM) model. Note that, both linear and nonlinear classification models are under the same empirical risk minimization (ERM) principle (Eq. \eqref{obj-classification}). Thus, the PRESTIGE mechanism can be readily leveraged by nonlinear classification model as well, if we represent $\mathcal{C}_w$ by deep neural networks.

\subsubsection{Composite Privacy} PRESTIGE preserves a composite privacy, using private curriculum and private sampling together. According to composition theorem in Lemma~\ref{composition}, if $\varepsilon = \varepsilon_r + \varepsilon_s$, PRESTIGE keeps $\varepsilon$-differentially local privacy.
\begin{lemma}\label{composition}
(Theorem 3.14 in \cite{dwork2014algorithmic}). Let $\mathcal{M}_1$ be an $\varepsilon_1$ differentially private algorithm, and let $\mathcal{M}_2$ be an $\varepsilon_2$ differentially private algorithm. Then their sequential combination, defined to be $\mathcal{M}_{1,2}$, is ($\varepsilon_1 + \varepsilon_2$)-differentially private.
\end{lemma}
\subsection{Theoretical Analysis}
Theorem~\ref{Convergence} reveals the convergence rate of PRESTIGE, which is similar to $\mathcal{O}(\frac{1}{\sqrt{T}})$. To obtain a sharper/faster convergence rate, we propose Corollary~\ref{corollary}. Theorem~\ref{LowerBound} explores the maximum complexity of PRESTIGE, which provides the formal termination criteria for early stopping. All theorems hold for the choice of any regularizer.
%
%

\subsubsection{Convergence Rate} We analyze the convergence rate of PRESTIGE, which demonstrates that our algorithm can converge. We use $\mathbb{E}\big[\cdot\big]$ to denote the \mbox{expectation}. Before delving into Theorem~\ref{Convergence}, we first present a fundamental lemma below.

\begin{lemma}\label{scalarbound}
If we leverage the Proposition~\ref{LocallyPrivateSampling} in Section \ref{Privacy-Preservation}, namely, we privatize the random variable $X \in \mathbb{R}^d$ into the variable $Z \in \mathbb{R}^d$ by private sampling, then to achieve the unbiasedness condition $\mathbb{E}[Z|X]=X$, the scalar bound $B$ should be set as:
\begin{equation}
B = L\sqrt{\pi}\frac{e^{\varepsilon_s}+1}{e^{\varepsilon_s}-1}\frac{d\Gamma (\frac{d-1}{2}+1)}{\Gamma(\frac{d}{2}+1)},
\end{equation}
where $\lVert X\rVert_2 \leq L$ and $\varepsilon_s$ is privacy parameter of private sampling.
\end{lemma}

\begin{proof}
We privatize the random variable $X \in \mathbb{R}^d$ into the variable $Z \in \mathbb{R}^d$ by private sampling with privacy parameter $\varepsilon_s$. Moreover, we assume that $X \in \mathbb{R}^d$ with the constraint of $\lVert X \rVert_2 \leq L$. Since the private sampling is highly related to $d$-dimensional ball, therefore, we first explore the surface area of $d$-dimensional ball. According to the constraint, the radius of this ball should be $L$. Thus, the surface area (denoted as $SA$) of $d$-dimensional ball with the radius $L$ is $SA_d(L) = \frac{dL^{d-1}\pi^{\frac{d}{2}}}{\Gamma(\frac{d}{2}+1)}$. Note that, for simplicity, we first set $L = 1$ to derive a simple result, and then generalize this result with arbitrary $L$. If we sample a random variable $V$ uniformly on the surface of this ball with $L = 1$, and assume the first coordinate $V_1 = e_1$, then we have:
\begin{equation}\label{coordinate}
E[V] = \frac{2e_1}{SA_d(1)}\int_{0}^{1}SA_{d-1}(\sqrt{1-k^2})kdk,
\end{equation}
where we leverage the symmetry of $d$-dimensional ball. To integrate the Eq. \eqref{coordinate} easier, we switch from cartesian coordinate system to polar coordinate system. Namely:
\begin{equation}
\begin{split}
E[V] &= \frac{2e_1}{SA_d(1)}\int_{0}^{1}SA_{d-1}(\sqrt{1-k^2})kdk\\
& = e_1\frac{2SA_{d-1}(1)}{SA_d(1)}\int_{0}^{\frac{\pi}{2}}\cos^{d-2}(\theta)\sin(\theta)d\theta\\
& = e_1\frac{2SA_{d-1}(1)}{SA_d(1)}\int_{0}^{\frac{\pi}{2}}(-\frac{\frac{d}{d\theta}\cos^{d-1}(\theta)}{d-1})d\theta\\
& = e_1\frac{2SA_{d-1}(1)}{SA_d(1)}\frac{1}{d-1}\\
& = e_1 2 \frac{\Gamma(\frac{d}{2}+1)}{d\pi^{\frac{d}{2}}}\frac{(d-1)\pi^{\frac{d-1}{2}}}{\Gamma(\frac{d-1}{2}+1)}\frac{1}{d-1}\\
& = e_1 2 \frac{\Gamma(\frac{d}{2}+1)}{d\sqrt{\pi}\Gamma(\frac{d-1}{2}+1)}.
\end{split}
\end{equation}
Note that, the procedure of locally-private sampling is uniformly processed on the half of $d$-dimensional ball. Due to $\lVert \mathbb{E}[Z]\rVert_2 = B$ and $\lVert X\rVert_2 = L$, therefore, we generalize $\mathbb{E}[Z|X]$ from coefficients of $\mathbb{E}[V]$ by leveraging the rotational symmetry and different sampling probability:
\begin{equation}
\mathbb{E}[Z|X] = \frac{X}{\lVert X\rVert_2}B\frac{\Gamma(\frac{d}{2}+1)}{d\sqrt{\pi}\Gamma(\frac{d-1}{2}+1)}(\frac{e^{\varepsilon_s}}{e^{\varepsilon_s} + 1} - \frac{1}{e^{\varepsilon_s} + 1}).
\end{equation}
Therefore, to achieve the unbiasedness condition $\mathbb{E}[Z|X]=X$, we have:
\begin{equation}
\frac{B}{L}\frac{\Gamma(\frac{d}{2}+1)}{d\sqrt{\pi}\Gamma(\frac{d-1}{2}+1)}(\frac{e^{\varepsilon_s}}{e^{\varepsilon_s} + 1} - \frac{1}{e^{\varepsilon_s} + 1}) = 1.
\end{equation}
Namely, we have:
\begin{equation}
B = L\sqrt{\pi}\frac{e^{\varepsilon_s}+1}{e^{\varepsilon_s}-1}\frac{d\Gamma (\frac{d-1}{2}+1)}{\Gamma(\frac{d}{2}+1)}.
\end{equation}
Then the proof completes.
\end{proof}

\begin{theorem}\label{Convergence}
For PRESTIGE, consider that any $\mathbf{w} \in \mathbb{B}_d(R)$ and the loss is $L$-Lipschitz with respect to the $l_p$-norm for some $p \in [2,\infty]$. At $t$-th iteration after randomized response, let $G_{p_t}$ be generated from stochastic gradient $g_t \in \mathbb{R}^d$ by private sampling with privacy parameter $\varepsilon_s$, and $G_{p_t}$ be restricted under the scalar bound $B = L\sqrt{\pi}\frac{e^{\varepsilon_s}+1}{e^{\varepsilon_s}-1}\frac{d\Gamma (\frac{d-1}{2}+1)}{\Gamma(\frac{d}{2}+1)}$ (Lemma~\ref{scalarbound}). Assume that $\mathbf{w^*}$ is defined as a local minimum, and the dynamic learning rate $\eta_t$ is monotonically decreasing with $\eta_t = \frac{\eta_0}{\lambda\sqrt{t}}$, where $\eta_0 = \frac{R}{B}$. Then, after $T$ actual updates, the convergence rate of PRESTIGE in expectation is
\begin{equation*}
\mathbb{E}\big[F(\hat{\mathbf{w}}^{(T)})\big] - F(\mathbf{w^*}) \leq \frac{c(\varepsilon_s,\lambda)k(d)}{\sqrt{T}},
\end{equation*}
where $F()$ denotes the classification model (Eq.~\eqref{obj-classification}), $\hat{\mathbf{w}}^{(T)} = \frac{1}{T}\sum_{t=1}^{T}\mathbf{w}^{(t)}$, $c(\varepsilon_s,\lambda) = \frac{(\lambda^2+2)(e^{\varepsilon_s}+1)}{2\lambda\varepsilon_s}$ is a dynamic function related to privacy level $\varepsilon_s$ and regularization intensity $\lambda$, and $k(d) = LR\sqrt{\pi}\sqrt{d}$ is a constant dependent of dimension $d$.
\end{theorem}
\begin{remark}
We conclude that when $T = \frac{c^2(\varepsilon_s,\lambda)k^2(d)}{\epsilon^2}$, PRESTIGE has $\epsilon$-solution~\footnote{Please refer to the page $47/315$ in KDD15 tutorial: \url{https://homepage.cs.uiowa.edu/~tyng/kdd15-tutorial.pdf} for a good visualization of $\epsilon$.}. Moreover, the convergence rate of PRESTIGE is $\mathcal{O}(\frac{\lambda e^{\varepsilon_s}\sqrt{d}}{\varepsilon_s\sqrt{T}})$, which is highly related to the regularization intensity $\lambda$, privacy parameter $\varepsilon_s$ and dimension $d$. Note that, the convergence rate of PRESTIGE can be further sped up with Nesterov’s accelerated strategy \cite{nesterov2007gradient,o2015adaptive}, which will be discussed in our future work.
\end{remark}

\begin{proof}
According to the update rule of PRESTIGE, $\mathbf{w}^{(t+1)} = \mathbf{\mathbf{w}}^{(t)} -\eta_t G_{p_t}$, where $G_{p_t}$ is generated from $g_t = \nabla f_{it}(\mathbf{w}^{(t)})$ by private sampling, and $t$ is the current number of actual updates varying from $1 \cdots T$. The random number $it$ belongs to the set $\{1,...,n\}$ while $z_{it}(\mathbf{w}^{(t-1)}) \geq D_{th}$. Since $G_{p_t}$ is restricted under the bound $B = L\sqrt{\pi}\frac{e^{\varepsilon_s}+1}{e^{\varepsilon_s}-1}\frac{d\Gamma (\frac{d-1}{2}+1)}{\Gamma(\frac{d}{2}+1)}$, therefore, we have $\mathbb{E}[G_{p_t}|g_t]=g_t$. Assume the number of training sample $n$ is very large, due to~\eqref{obj-classification}, $\lvert \mathbb{E}\big[\nabla f_{it}(\mathbf{w}^{(t)})\big] - \nabla F(\mathbf{w}^{(t)})\rvert \leq \epsilon$. With the increase of actual updates, the decrease of dynamic threshold $D_{th}$ makes $\epsilon$ monotonically decrease towards $0$. Thus, we construct the following inequality:
\begin{equation}\label{original}
\begin{split}
&\quad\quad\frac{1}{2}\lVert \mathbf{w}^{(t+1)} - \mathbf{w^*} \rVert_2^2\\
& \leq \frac{1}{2} \lVert \mathbf{w}^{(t)} - \mathbf{w^*} \rVert_2^2 + \frac{{\eta_t}^2}{2}\lVert G_{p_t} \rVert_2^2 - \eta_t \langle G_{p_t}, \mathbf{w}^{(t)} - \mathbf{w^*} \rangle.
\end{split}
\end{equation}

Due to the characteristic of first-order convexity, we have:
\begin{equation}\label{1st-order}
\langle \nabla F(\mathbf{w}^{(t)}), \mathbf{w}^{(t)} - \mathbf{w^*} \rangle \geq F(\mathbf{w}^{(t)}) - F(\mathbf{w}^{*}).
\end{equation}

If we multiply $-\eta_t$ to both hand side of Eq. \eqref{1st-order}, we have:
\begin{equation}
-\eta_t \langle \nabla F(\mathbf{w}^{(t)}), \mathbf{w}^{(t)} - \mathbf{w^*} \rangle \leq -\eta_t (F(\mathbf{w}^{(t)}) - F(\mathbf{w}^{*})).
\end{equation}

Therefore, Eq. \eqref{original} can be further calculated as below when we define $\epsilon_t = G_{p_t} - \nabla F(\mathbf{w}^{(t)})$:
\begin{equation}\label{private-eq}
\begin{split}
&\quad\quad\frac{1}{2}\lVert \mathbf{w}^{(t+1)} - \mathbf{w^*} \rVert_2^2\\
& \leq \frac{1}{2} \lVert \mathbf{w}^{(t)} - \mathbf{w^*} \rVert_2^2 + \frac{{\eta_t}^2}{2}\lVert G_{p_t} \rVert_2^2 -\eta_t (F(\mathbf{w}^{(t)}) - F(\mathbf{w}^{*}))\\
& \quad -\eta_t \langle \epsilon_t, \mathbf{w}^{(t)} - \mathbf{w^*} \rangle.
\end{split}
\end{equation}

If we exchange items of Eq. \eqref{private-eq}, namely, $\frac{1}{2}\lVert \mathbf{w}^{(t+1)} - \mathbf{w^*} \rVert_2^2$ and $\eta_t(F(\mathbf{w}^{(t)}) - F(\mathbf{w}^{*}))$ we have:
\begin{equation}\label{private-eq-rev}
\begin{split}
&\quad\quad F(\mathbf{w}^{(t)}) - F(\mathbf{w}^{*}) \\
& \leq \frac{1}{2\eta_t} (\lVert \mathbf{w}^{(t)} - \mathbf{w^*} \rVert_2^2 - \lVert \mathbf{w}^{(t+1)} - \mathbf{w^*} \rVert_2^2) + \frac{\eta_t}{2}\lVert G_{p_t} \rVert_2^2\\
& \quad - \langle \epsilon_t, \mathbf{w}^{(t)} - \mathbf{w^*} \rangle.
\end{split}
\end{equation}

If we sum up the both hand side of Eq. \eqref{private-eq-rev} from $1 \cdots T$, then we have:
\begin{equation}\label{sum-eq}
\begin{split}
&\quad\quad \sum_{t=1}^{T}\big[F(\mathbf{w}^{(t)}) - F(\mathbf{w}^{*})\big] \\
& \leq \sum_{t=1}^{T}\{\frac{1}{2\eta_t} (\lVert \mathbf{w}^{(t)} - \mathbf{w^*} \rVert_2^2 - \lVert \mathbf{w}^{(t+1)} - \mathbf{w^*} \rVert_2^2)\}\\
& \quad  + \frac{1}{2}\sum_{t=1}^{T}\eta_t \lVert G_{p_t} \rVert_2^2 - \sum_{t=1}^{T}\langle \epsilon_t, \mathbf{w}^{(t)} - \mathbf{w^*} \rangle.
\end{split}
\end{equation}

Note that, since $\eta_t$ is monotonically decreasing, namely $\eta_t \geq \eta_{t+1}$, then we have $\frac{1}{\eta_1} \leq \cdots \leq \frac{1}{\eta_{T}}$. Based on this result, we arrange:

\begin{equation}\label{basic-eq}
\begin{split}
&\quad\quad \sum_{t=1}^{T}\{\frac{1}{2\eta_t} (\lVert \mathbf{w}^{(t)} - \mathbf{w^*} \rVert_2^2 - \lVert \mathbf{w}^{(t+1)} - \mathbf{w^*} \rVert_2^2)\}\\
&\leq \frac{1}{2\eta_T}\lVert \mathbf{w}^{(1)} - \mathbf{w^*} \rVert_2^2 \leq \frac{1}{2\eta_T}\lVert \mathbf{w}^{(1)}\rVert_2^2 \leq \frac{1}{2\eta_T}R^2.
\end{split}
\end{equation}

Based on Eq. \eqref{basic-eq}, we transform Eq. \eqref{sum-eq} into:
\begin{equation}\label{sum-eq-2}
\begin{split}
&\quad\quad \sum_{t=1}^{T}\big[F(\mathbf{w}^{(t)}) - F(\mathbf{w}^{*})\big] \\
& \leq \frac{1}{2\eta_T}R^2  + \frac{1}{2}\sum_{t=1}^{T}\eta_t \lVert G_{p_t} \rVert_2^2 - \sum_{t=1}^{T}\langle \epsilon_t, \mathbf{w}^{(t)} - \mathbf{w^*} \rangle.
\end{split}
\end{equation}

Due to the strategy of local privacy preservation, we have $\lVert G_{p_t} \rVert_2^2 = B$. Then, we take expectation for Eq. \eqref{sum-eq-2}:
\begin{equation}
\mathbb{E}\{\sum_{t=1}^{T}[F(\mathbf{w}^{(t)}) - F(\mathbf{w}^{*})]\} \leq \frac{1}{2\eta_T}R^2  + \frac{1}{2}\sum_{t=1}^{T}\eta_t B^2,
\end{equation}

where $\mathbb{E}[\sum_{t=1}^{T}\langle \epsilon_t, \mathbf{w}^{(t)} - \mathbf{w^*} \rangle] = 0$ because the following derivations:

\begin{equation}
\mathbb{E}[\langle \epsilon_t, \mathbf{w}^{(t)} - \mathbf{w^*} \rangle] = \mathbb{E}[\langle g_t - \nabla F(\mathbf{w}^{(t)}), \mathbf{w}^{(t)} - \mathbf{w^*} \rangle] = 0.
\end{equation}

If we set $\eta_t = \frac{\eta_0}{\lambda\sqrt{t}} = \frac{R}{\lambda B\sqrt{t}}$ to make sure that $\eta_t$ is monotonically decreasing, we have
\begin{equation}\label{core-eq}
\mathbb{E}\{\frac{1}{T}\sum_{t=1}^{T}[F(\mathbf{w}^{(t)}) - F(\mathbf{w}^{*})]\} \leq \frac{1}{2T\frac{R}{\lambda B\sqrt{T}}}R^2  + \frac{B^2}{2T}\frac{R}{\lambda B}\sum_{t=1}^{T}\frac{1}{\sqrt{t}}.
\end{equation}

Due to the following golden rule:
\begin{equation}
\sum_{t=1}^{T}\frac{1}{\sqrt{t}} \leq \int_{0}^{T} x^{-\frac{1}{2}}dx = 2\sqrt{T}.
\end{equation}

Now Eq. \eqref{core-eq} can be transformed into:
\begin{equation}
\mathbb{E}[F(\hat{\mathbf{w}}^{(T)}) - F(\mathbf{w}^{*})] \leq \frac{\lambda RB}{2\sqrt{T}} + \frac{RB}{\lambda\sqrt{T}} = \frac{\lambda^2+2}{2\lambda}\frac{RB}{\sqrt{T}}.
\end{equation}

According to Appendix F.2 in \cite{duchi2013local-arXiv}, to achieve $\mathbb{E}[G_{p_t}|g_t] = g_t$, we should set $B = L\sqrt{\pi}\frac{e^{\varepsilon_s}+1}{e^{\varepsilon_s}-1}\frac{d\Gamma (\frac{d-1}{2}+1)}{\Gamma(\frac{d}{2}+1)}$. Since $\frac{d\Gamma (\frac{d-1}{2}+1)}{\Gamma(\frac{d}{2}+1)} \leq \sqrt{d}$ and $e^{\varepsilon_s} - 1 = \varepsilon_s + \mathcal{O}(\varepsilon_s^2)$ (Taylor series of exponential function), we have:

\begin{equation}
\begin{split}
B &= L\sqrt{\pi}\frac{e^{\varepsilon_s}+1}{e^{\varepsilon_s}-1}\frac{d\Gamma (\frac{d-1}{2}+1)}{\Gamma(\frac{d}{2}+1)}\\
&\leq L\sqrt{\pi}\frac{e^{\varepsilon_s}+1}{\varepsilon_s}\sqrt{d} = \frac{\sqrt{d}}{\varepsilon_s}(e^{\varepsilon_s}+1)L\sqrt{\pi}.
\end{split}
\end{equation}

Therefore, we have:
\begin{equation}
\begin{split}
& \quad\quad \mathbb{E}[F(\hat{\mathbf{w}}^{(T)}) - F(\mathbf{w}^{*})]\\
& \leq \frac{\lambda^2+2}{2\lambda}\frac{RB}{\sqrt{T}}\\
& \leq \frac{\lambda^2+2}{2\lambda} \frac{R}{\sqrt{T}}\frac{\sqrt{d}}{\varepsilon_s}(e^{\varepsilon_s}+1)L\sqrt{\pi}\\
& = \frac{(\lambda^2+2)(e^{\varepsilon_s}+1)}{2\lambda\varepsilon_s}\frac{LR\sqrt{\pi}\sqrt{d}}{\sqrt{T}}\\
& = \frac{c(\varepsilon_s,\lambda)k(d)}{\sqrt{T}}.
\end{split}
\end{equation}
where $c(\varepsilon_s,\lambda) = \frac{(\lambda^2+2)(e^{\varepsilon_s}+1)}{2\lambda\varepsilon_s}$ is a dynamic function related to privacy level $\varepsilon_s$ and regularization intensity $\lambda$, and $k(d) = LR\sqrt{\pi}\sqrt{d}$ is a constant dependent of $d$ dimension.
\end{proof}

\begin{corollary}\label{corollary}
Suppose that $q$ conjugates to $p$ and satisfies that $\frac{1}{p} + \frac{1}{q} = 1$. If we restrict $\mathbf{w} \in \mathbb{B}_d(Rd^{\frac{1}{2}-\frac{1}{q}})$, due to $p \in [2,\infty]$, then we have $\lVert \mathbf{w} \rVert_2 \leq Rd^{\frac{1}{2}-\frac{1}{q}} \leq R$. Therefore, the convergence rate of PRESTIGE in expectation can be sharper as follows:
\begin{equation}
\mathbb{E}\big[F(\hat{\mathbf{w}}^{(T)})\big] - F(\mathbf{w^*}) \leq \frac{c(\varepsilon_s,\lambda)k'(d)}{\sqrt{T}},
\end{equation}
where $\hat{\mathbf{w}}^{(T)} = \frac{1}{T}\sum_{t=1}^{T}\mathbf{w}^{(t)}$, $c(\varepsilon_s,\lambda) = \frac{(\lambda^2+2)(e^{\varepsilon_s}+1)}{2\lambda\varepsilon_s}$, and $k'(d) = LR\sqrt{\pi}d^{1-\frac{1}{q}}$ is a reduced constant, and \mbox{$k'(d) < k(d)$}.
\end{corollary}

\begin{remark}
The result of Corollary 1 can be also derived in parallel by Eq. (11) of \cite{agarwal2010information} coupled with $\varepsilon_s$-differentially local privacy. Namely, the general minimax rate for the convergence rate is $\frac{LR\sqrt{\pi}d^{\frac{1}{2}-\frac{1}{q}}}{\sqrt{T}}$. The price for $\varepsilon_s$-differentially local privacy is a multiplicative factor ``$c(\varepsilon_s,\lambda)\sqrt{d}$''. Therefore, if we combine the result of general minimax rate with this factor, we can also derive the Corollary~\ref{corollary}.
\end{remark}

\subsubsection{Maximum Complexity} To provide the formal termination criteria for early stopping, we leverage the learning framework of KWIK (knows what it knows) \cite{li2008knows} to explore the maximum complexity ($T_{\max}$) of PRESTIGE. KWIK framework combines elements of the well-known Probably Approximately Correct (PAC) \cite{valiant1984theory} and Mistake-Bound (MB) \cite{littlestone1988learning} models. According our observation, PRESTIGE can be viewed as a curriculum KWIK algorithm. Specifically, in Algorithm~\ref{algorithm-PRESTIGE}, the input of PRESTIGE is $\mathcal{D} = \{\mathbf{x}_i, y_i\}_{i=1}^n$, and the output is the updated $\mathbf{w}$. Then, PRESTIGE is denoted as the predicted function $h$ in the hypothesis class $\mathcal{H}$, where $\mathcal{H} \subseteq (\mathcal{D} \rightarrow \mathbf{w})$. Before delving into Theorem~\ref{LowerBound}, we adapt two definitions by KWIK protocol.

\begin{definition}\label{Accuracy-Requirement}
(Accuracy Requirement). Given any updated $\mathbf{w} \in \mathbb{B}_d(R)$. Assume that $\mathbf{w^*}$ is defined as a local minimum. Therefore, $\mathbf{w}$ in the updated area must be $R$-accurate, that is, $\lVert \mathbf{w} - \mathbf{w^*} \rVert_2 \leq R$ whenever $z \geq D_{th}$.
\end{definition}

\begin{definition}\label{Sample-Complexity-Requirement}
(Sample Complexity Requirement). The virtual complexity
of PRESTIGE in the infeasible area ($z \leq D_{th}$), is bounded by a function polynomial in $1/R$, $1/\delta$, and dim($\mathcal{H}$), where dim($\mathcal{H}$) measures the dimension or complexity of the hypothesis class $\mathcal{H}$.
\end{definition}

\begin{theorem}\label{LowerBound} Given the complexity ratio $r$ between maximum and virtual complexity of PRESTIGE, $\mathcal{H}$ is efficiently KWIK-learnable if PRESTIGE satisfies Definition~\ref{Accuracy-Requirement} and \ref{Sample-Complexity-Requirement} simultaneously with probability at least $1-\delta$ ($0 < \delta < 1$). Then, the maximum complexity of PRESTIGE, $T_{\max}$, is polynomial in $r$, $1/R$, $1/\delta$, and dim($\mathcal{H}$).
\end{theorem}

\begin{remark}
PRESTIGE can be viewed as a curriculum KWIK algorithm with a KWIK bound $\mathcal{O}(\frac{r}{R}\ln\frac{1}{\delta})$.
\end{remark}

\section{Numerical Experiments}\label{Experiments}
In this section, we conduct experiments to verify the robustness of PRESTIGE. Meanwhile, we also explore the effectiveness of mini-batch PRESTIGE, and the efficacy of PRESTIGE under different losses.


\subsection{Experimental Setup}\label{Experimental Testbed and Setup}
\begin{table}
\caption{Datasets used in this paper.}
\label{datasets}
\begin{center}
\scalebox{1}{
\begin{tabular}{llll}\hline
\multicolumn{1}{c}{\bf DATA SET}  &\multicolumn{1}{c}{\bf TRAINING PTS.} &\multicolumn{1}{c}{\bf TESTING PTS.}
&\multicolumn{1}{c}{\bf FEATURES.}\\ \hline
\textit{A7A} &16,100 &16,461 &123\\
\textit{MNIST38} &450,000 &97,570 &784\\
\textit{SUSY} &4,000,000 &1,000,000 &18\\
\textit{REAL-SIM} &57,847 &14,462 &20,985\\ \hline\hline
\textit{Cyberbully} &100,000 &20,000 &4,800\\
\textit{Diabetes} &17,288 &5,327 &4,090\\ \hline
\end{tabular}}
\end{center}
\end{table}
\subsubsection{Baselines}
There are three sets of baselines for different purposes. The first set consists of vanilla SGD and DJW, where DJW denotes Duchi-Jordan-Wainwright's model using private sampling \cite{duchi2013local-focs}. When we compare them with PRESTIGE, we can verify the robustness of PRESTIGE under the privacy preservation. Note that, vanilla SGD does not preserve any privacy, and cannot be directly used in any sensitive data. However, the performance of vanilla SGD should be better than that of privacy-preserving stochastic optimizations, as these optimizations inject the noise more or less. The second category consists of mini-batch SGD and mini-batch DJW. When we compare them with mini-batch PRESTIGE, we can further verify the effectiveness of mini-batch trick. The third category consists of vanilla SGD and DJW under different losses (e.g., a) convex and smooth loss; b) convex and non-smooth loss; c) non-convex and smooth loss; d) non-convex and non-smooth loss). When we compare them with PRESTIGE under different losses, we explore which loss benefits the robustness.
\subsubsection{Parameters\&Metrics}
For all stochastic methods, the parameter $\lambda$ is selected using $10$-fold cross validation (CV) in the range of \{$10^{-3},\cdots,10^{3}$\}, the maximum number of epochs $T_{max}$ is set to $10$ by Theorem~\ref{LowerBound}, and the primal variable $\mathbf{w}$ is initialized randomly. To preserve the privacy, we follow the settings from \cite{song2013stochastic,zhang2016privtree}, and set the total privacy budget \mbox{$\varepsilon \times n\times T_{max}$} in the range of \{$0.8,1.0,1.2,1.4,1.6$\}. For PRESTIGE, we set $\varepsilon = \varepsilon_r + \varepsilon_s$, where \mbox{$\varepsilon_r:\varepsilon_s = 1:4$}. The step size $\mu$ is empirically set to $1$. We repeat all experiments $20$ times. Then each plot is averaged across $20$ results. Since this paper aims to introduce a private but robust version of SGD mechanism, we evaluate its robustness by testing error rate with standard deviation.

\subsubsection{Miscellaneous}
There are three points to be noted that: (1) In the first and second sets of baselines, the default loss is the hinge loss. In the third set of baselines, we choose four different losses: logistic loss (convex and smooth), hinge loss (convex and non-smooth), Gompertz loss \cite{han2016convergence} (non-convex and smooth), and ramp loss \cite{collobert2006trading} (non-convex and non-smooth). For Gompertz loss, the parameter $c^*$ is set to~$2$. For ramp loss, the parameter $s^*$ is chosen in the range of $[-2,0]$. (2) PRESTIGE is to improve the robustness of DJW. According to the idea of control variable, the fair comparison should be between DJW and the curriculum version of DJW (PRESTIGE). (3) Experiments are implemented by Python on a cluster node with a 2.40GHz CPU and 32GB memory.

\subsection{Empirical Study}
\begin{figure}[!tp]
\center
\begin{tabular}{c|c}
\includegraphics[width=0.5\textwidth]{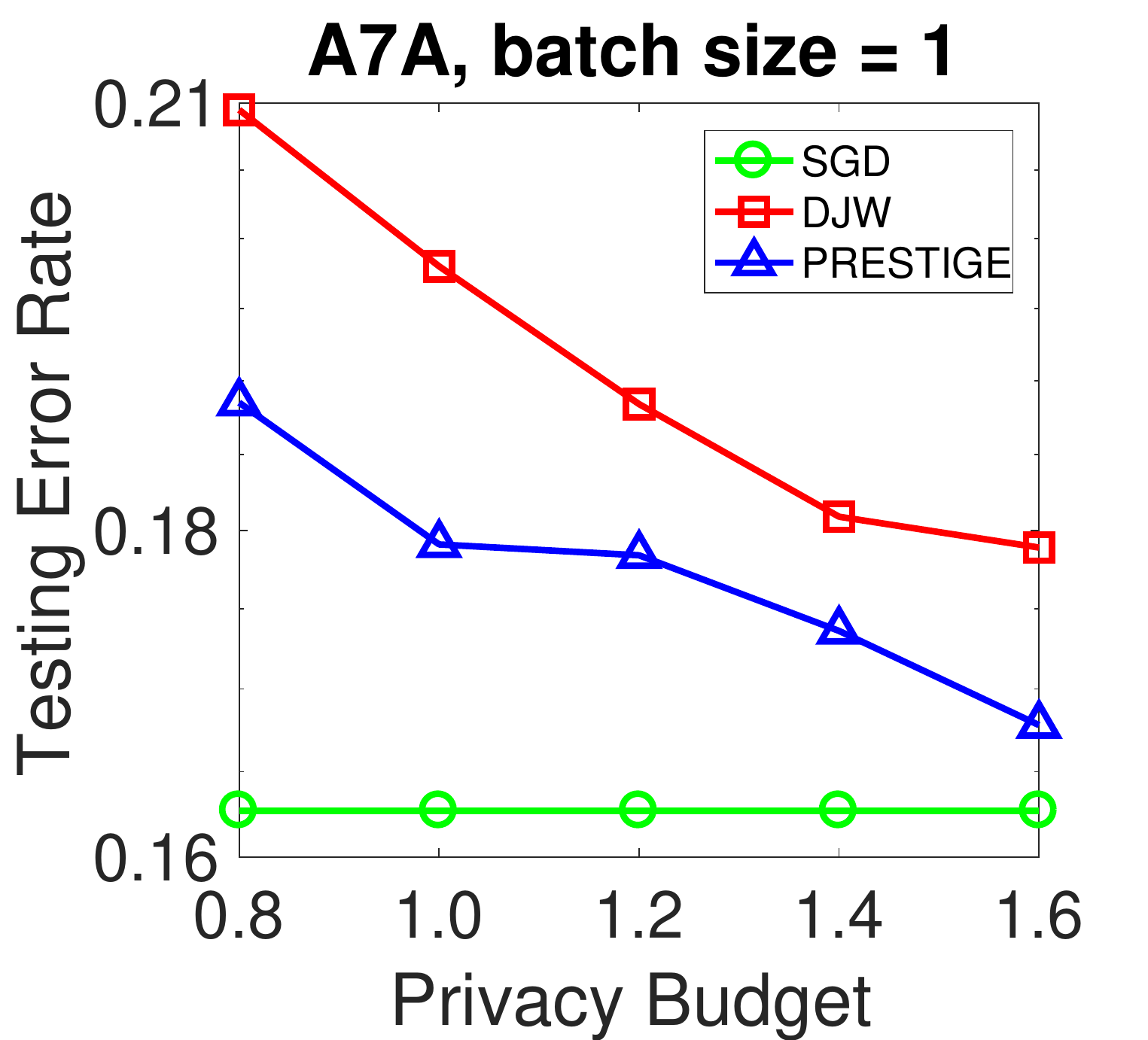} &
\includegraphics[width=0.5\textwidth]{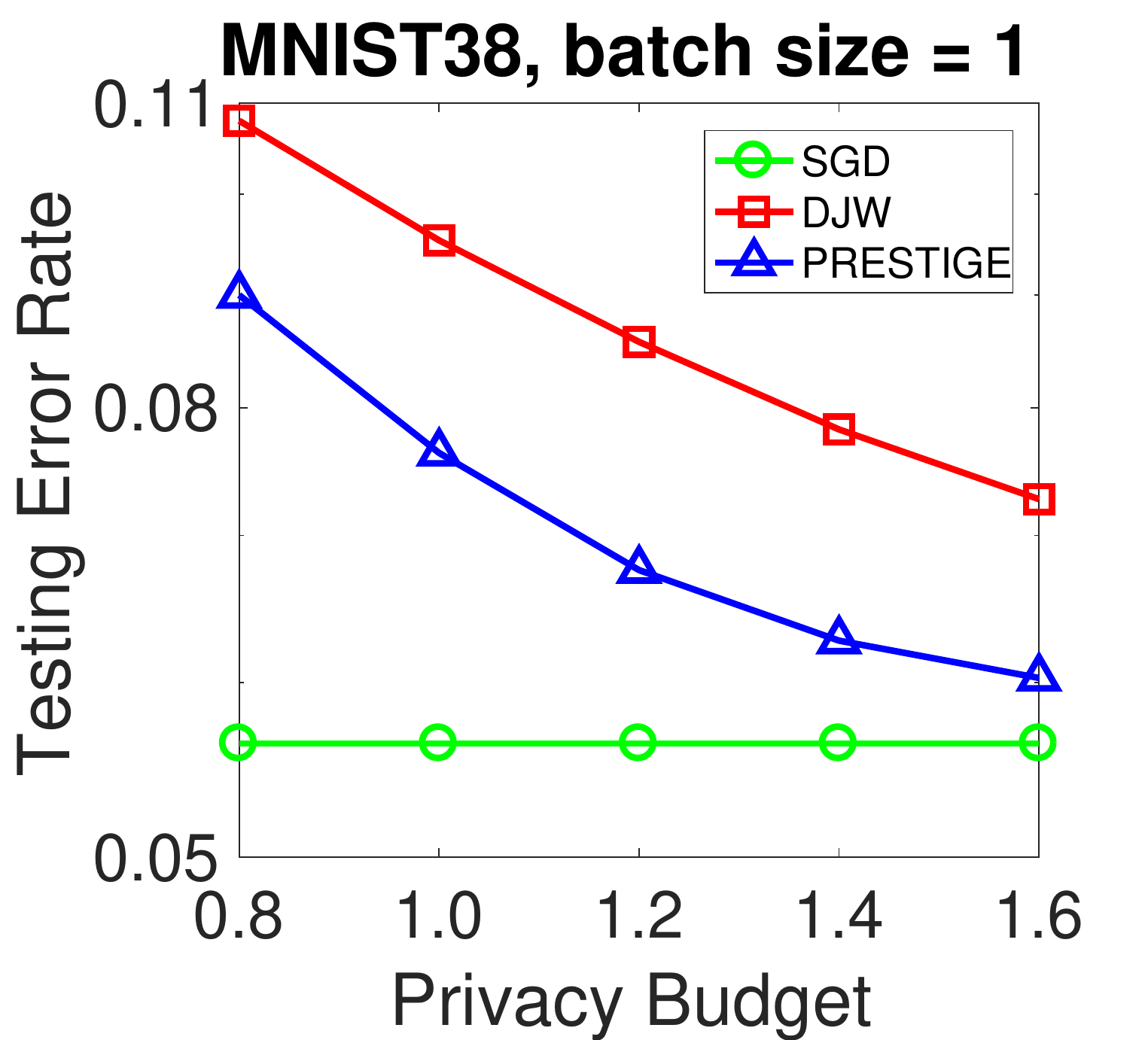} \\
\includegraphics[width=0.5\textwidth]{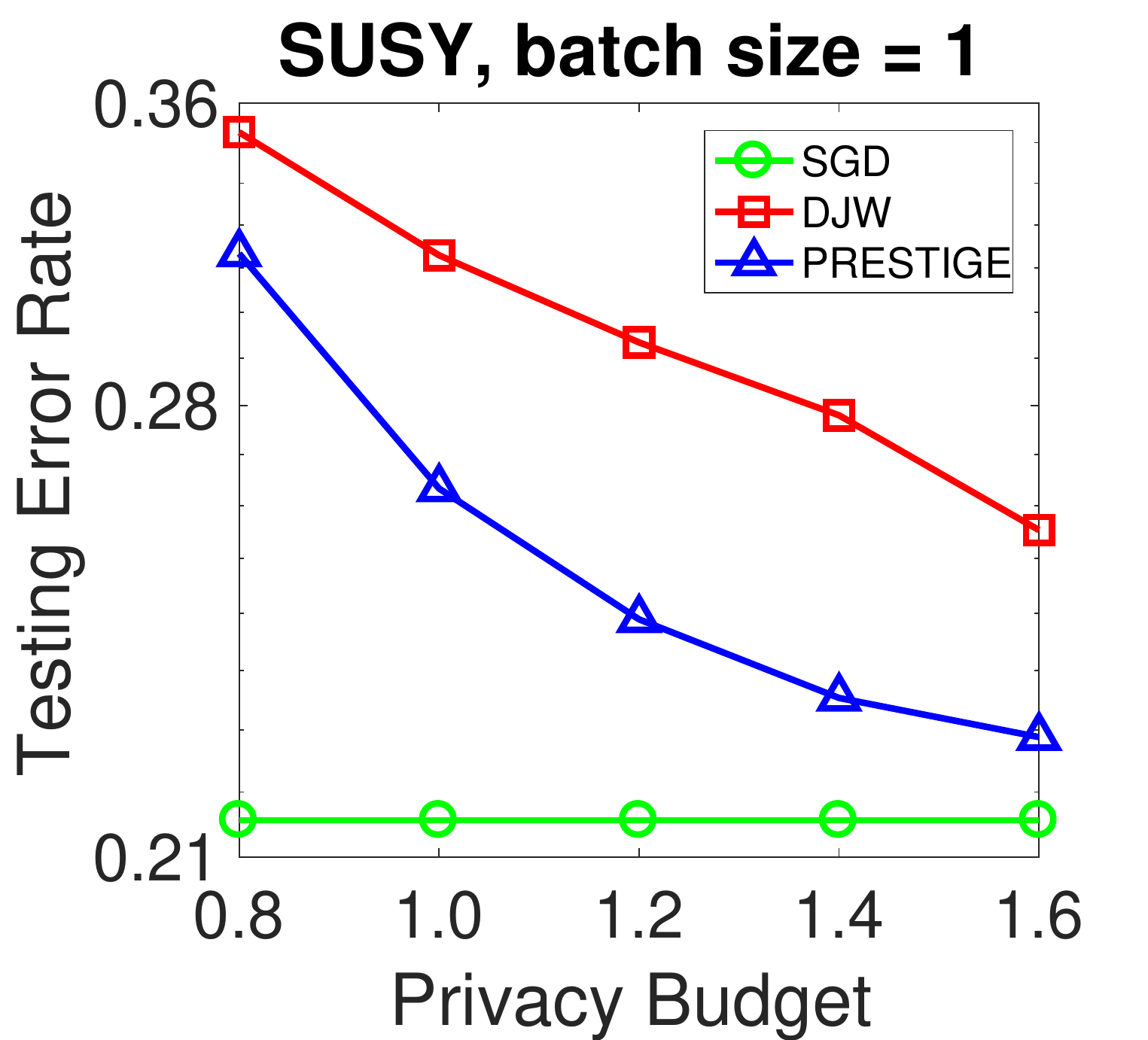} &
\includegraphics[width=0.5\textwidth]{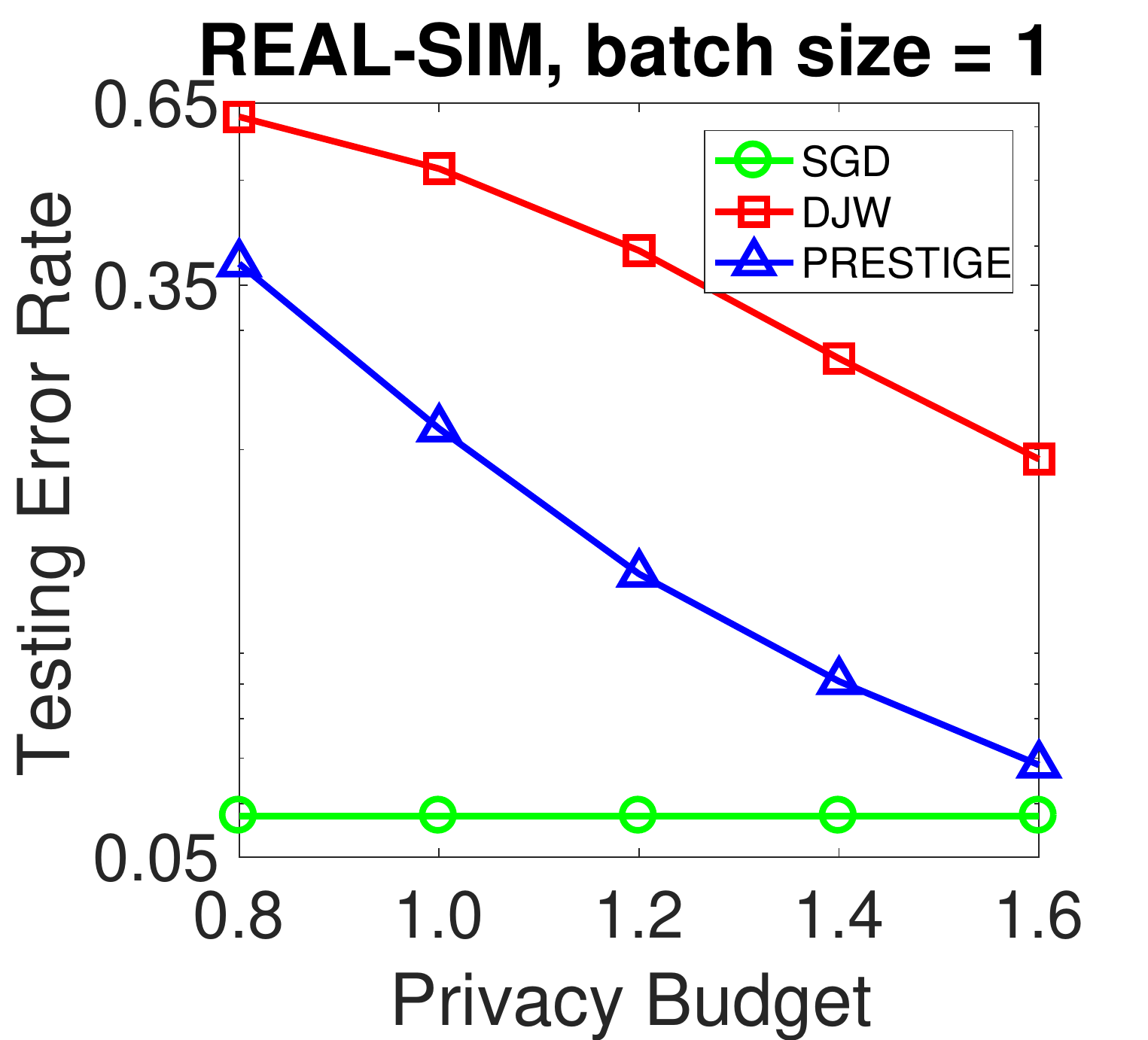}
\end{tabular}
\caption{To verify the \emph{robustness of PRESTIGE}, we compare PRESTIGE with vanilla SGD and DJW. We provide the testing error rate with the privacy budget $\varepsilon$ on four UCI datasets: small-scale \textit{A7A}, middle-scale \textit{MNIST38}, large-scale \textit{SUSY} and high-dimensional \textit{REAL-SIM}.}
\label{Robust-UCI}
\end{figure}

\begin{figure}[!tp]
\center
\begin{tabular}{c|c}
\includegraphics[width=0.5\textwidth]{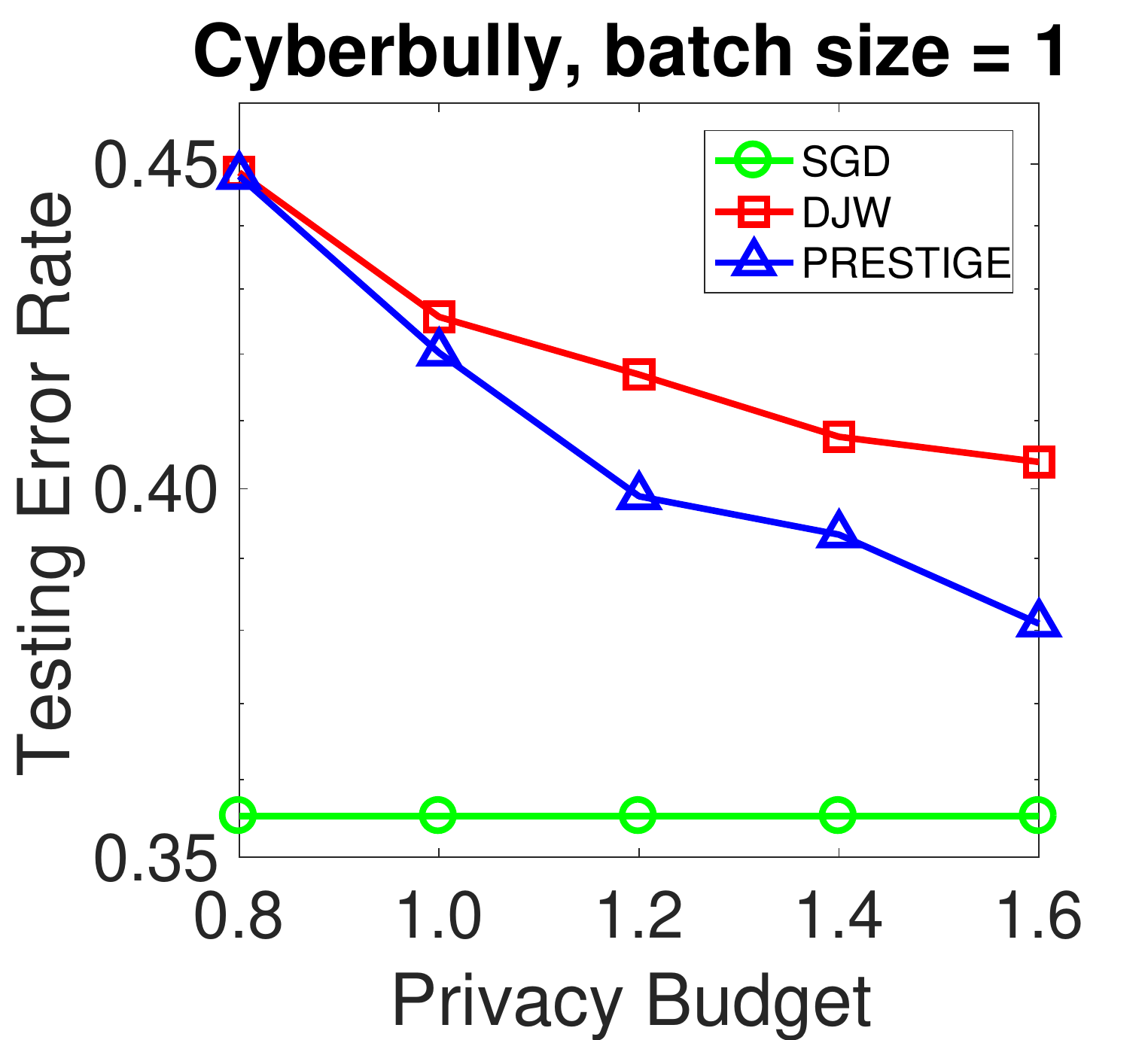} &
\includegraphics[width=0.5\textwidth]{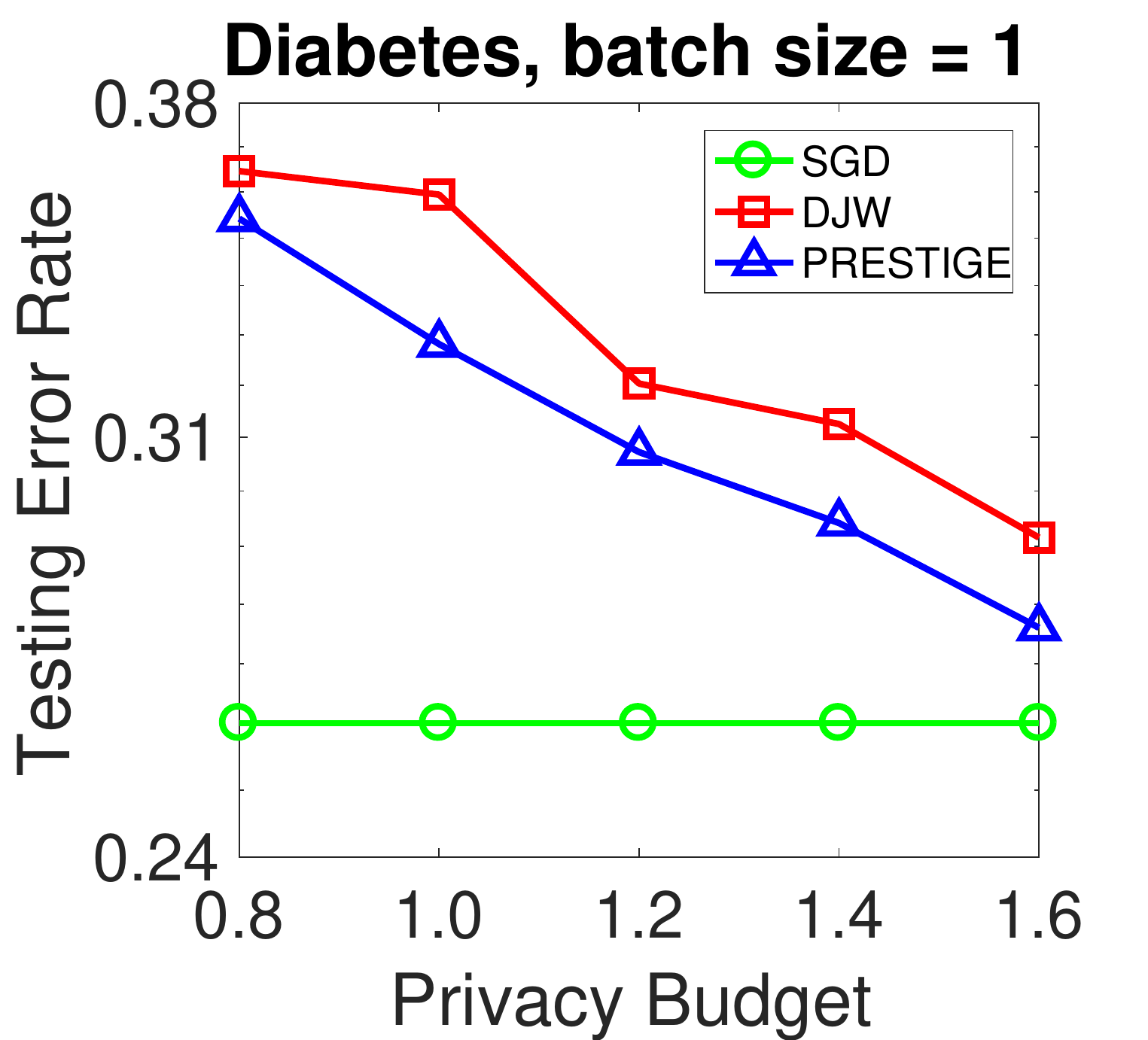}
\end{tabular}
\caption{To verify the \emph{robustness of PRESTIGE} in real-world situations, we compare PRESTIGE with vanilla SGD and DJW. We provide the testing error rate with the privacy budget $\varepsilon$ on two real-world datasets: \textit{Cyberbully} detection (social networks) and \textit{Diabetes} prediction (healthcare).}
\label{Robust-Real}
\end{figure}

\begin{figure}[!tp]
\center
\begin{tabular}{c|c}
\includegraphics[width=0.5\textwidth]{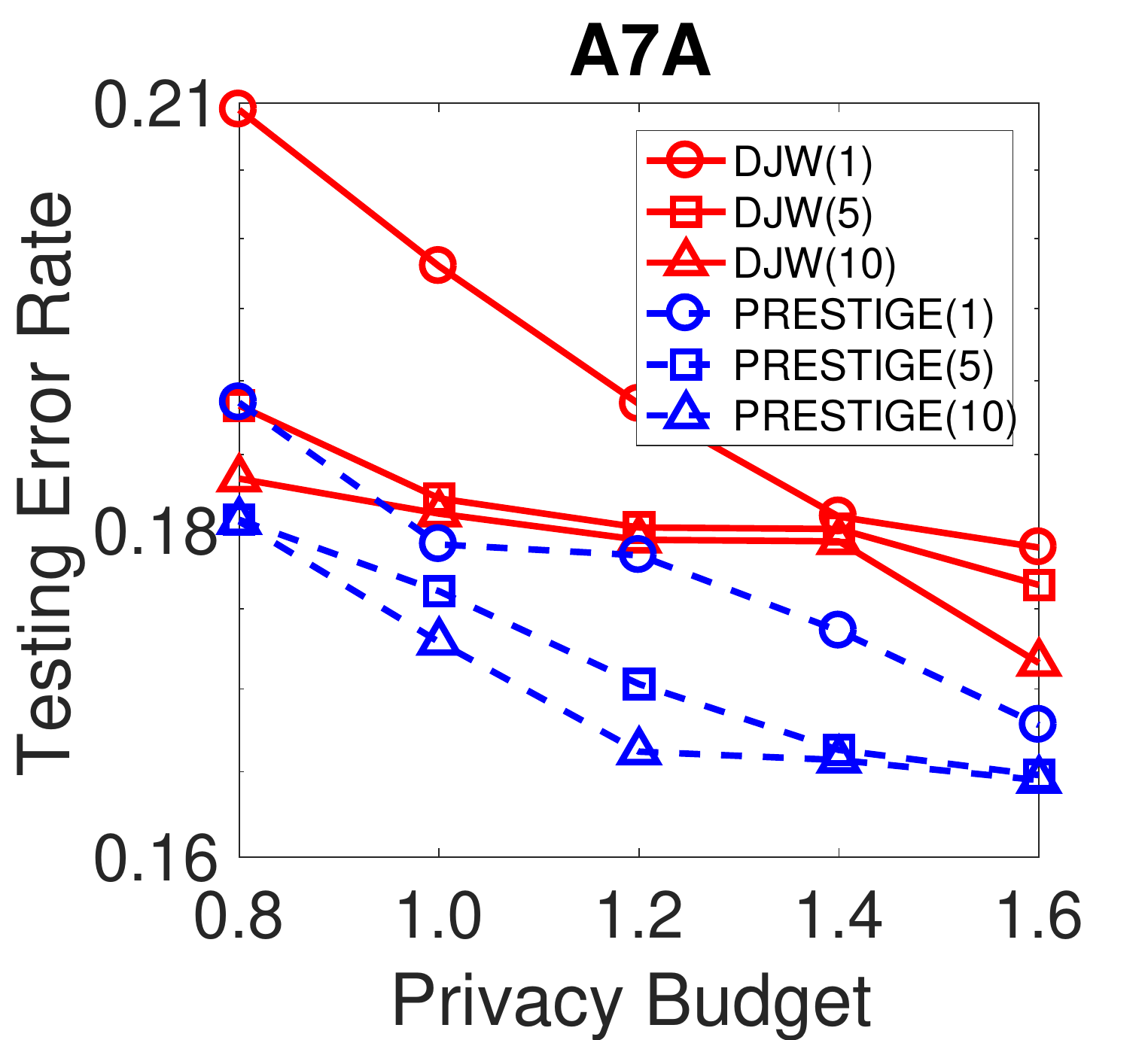} &
\includegraphics[width=0.5\textwidth]{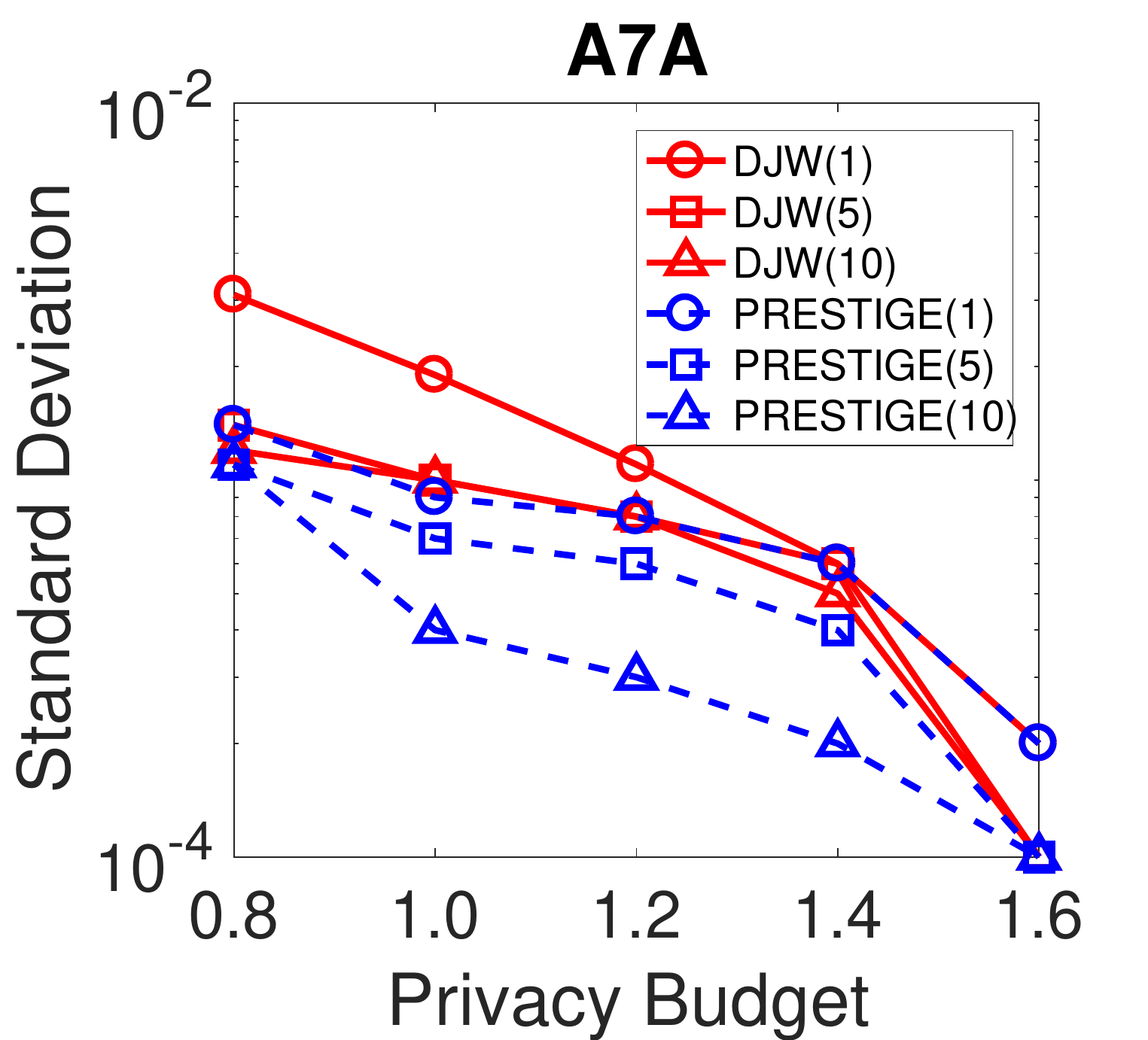} \\
\includegraphics[width=0.5\textwidth]{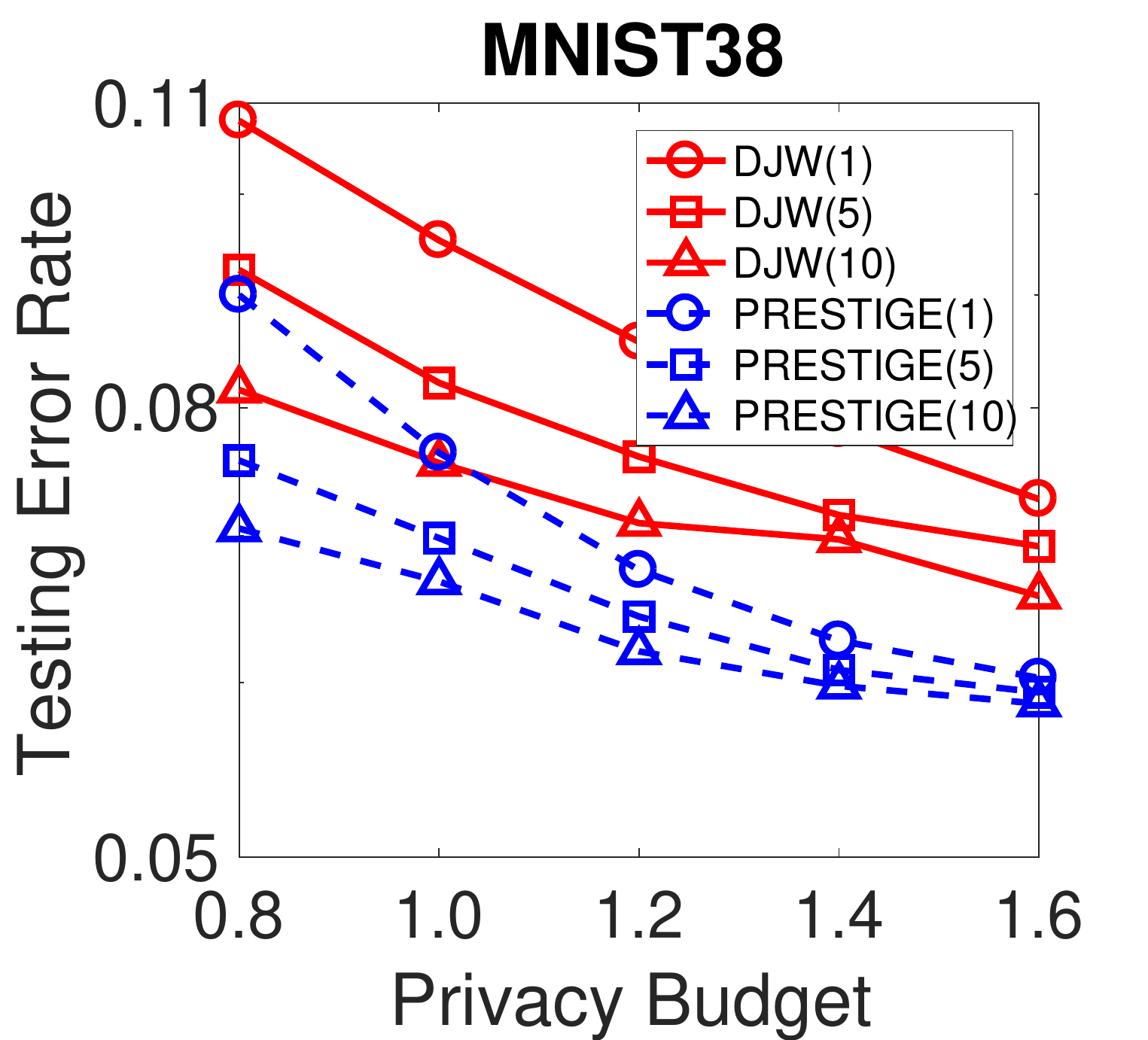} &
\includegraphics[width=0.5\textwidth]{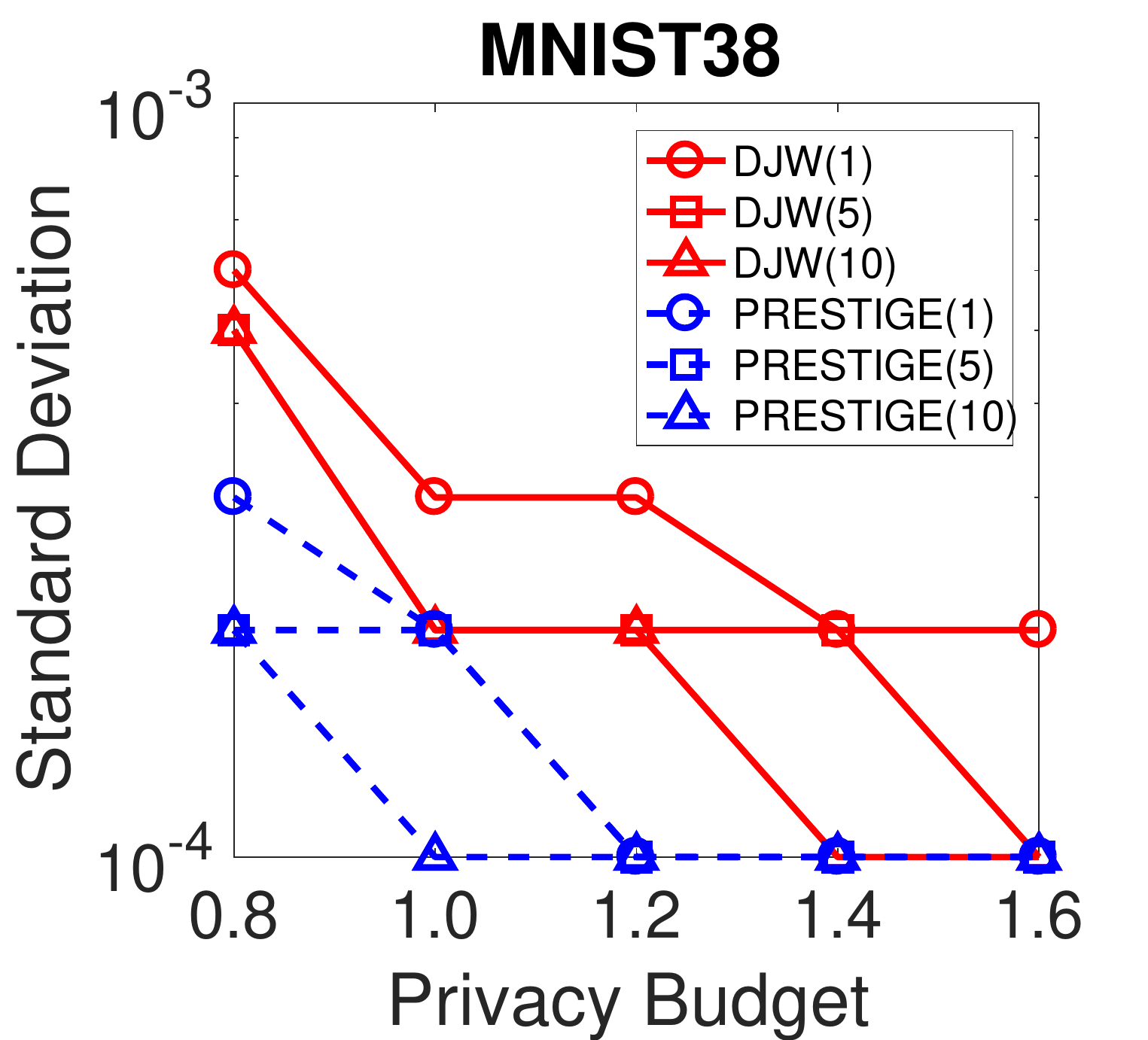} \\
\includegraphics[width=0.5\textwidth]{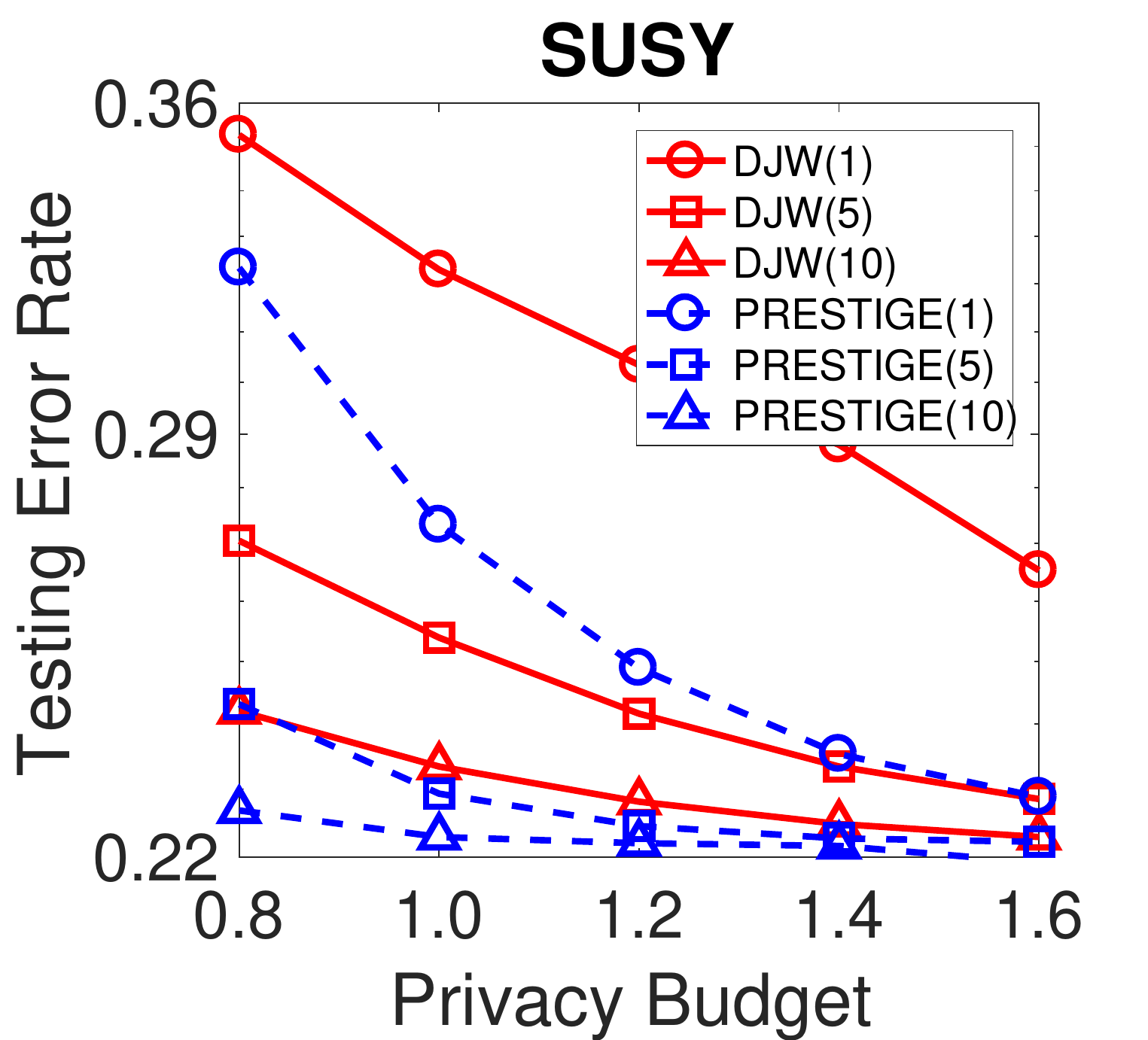} &
\includegraphics[width=0.5\textwidth]{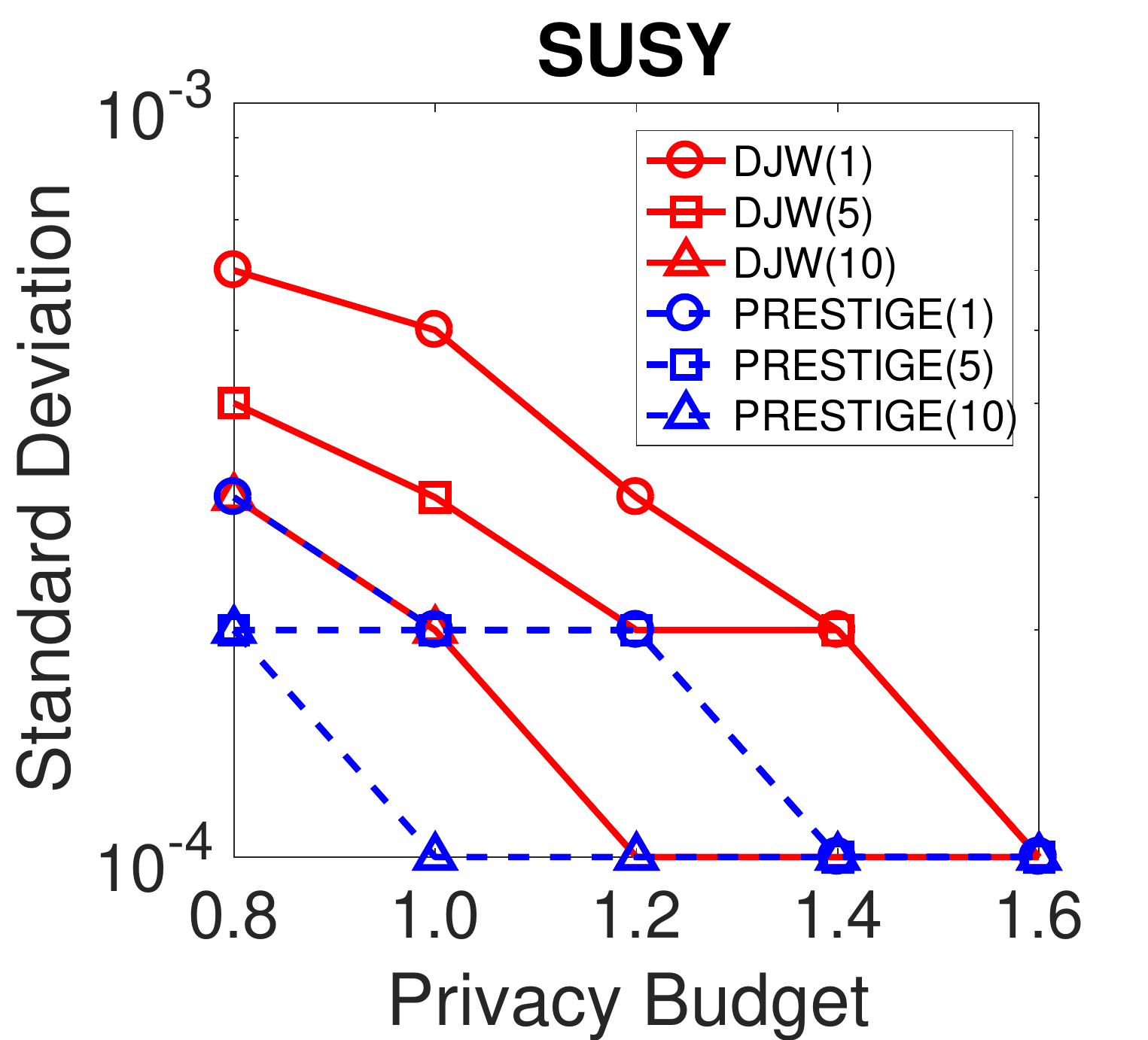} \\
\includegraphics[width=0.5\textwidth]{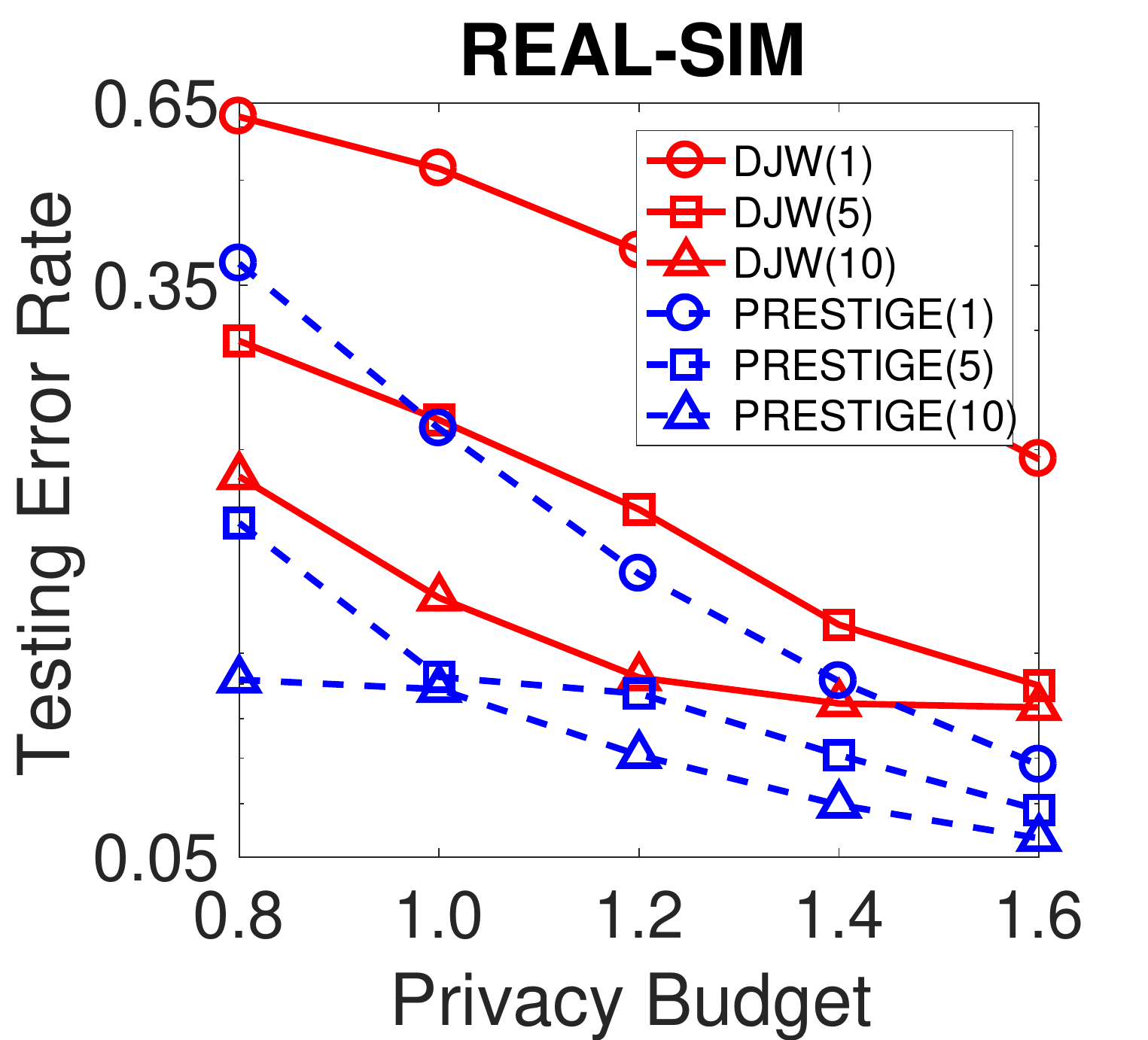} &
\includegraphics[width=0.5\textwidth]{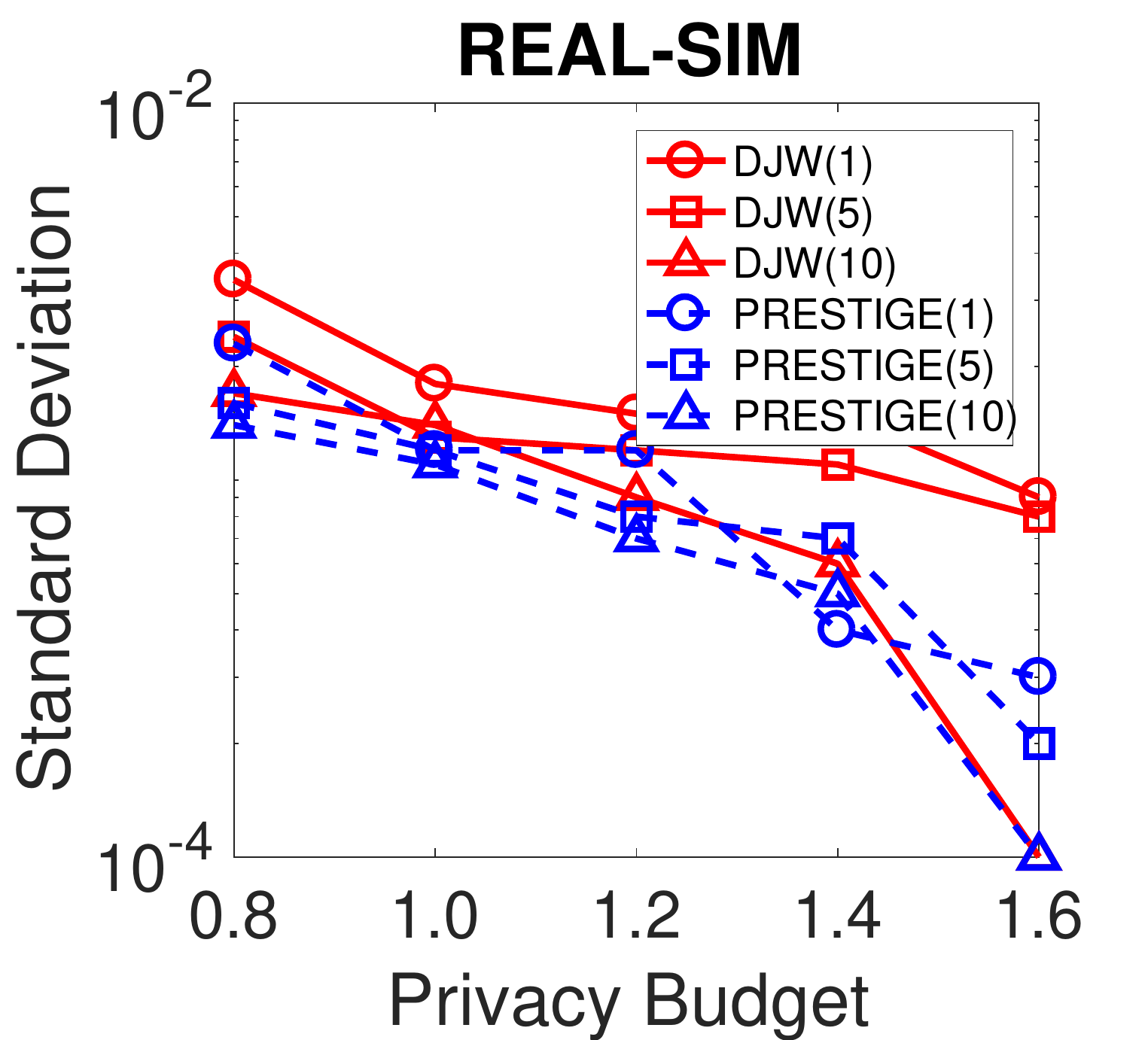}
\end{tabular}
\caption{To verify the \emph{effectiveness of mini-batch trick}, we compare the mini-batch version of PRESTIGE and DJW, where the batch size (listed in the parentheses) is $1$, $5$ and $10$ respectively. We provide the testing error rate and standard deviation with the privacy budget $\varepsilon$ on four UCI datasets.}
\label{Effectiveness-MiniBatch-UCI}
\end{figure}

\begin{figure}[!tp]
\center
\begin{tabular}{c|c}
\includegraphics[width=0.5\textwidth]{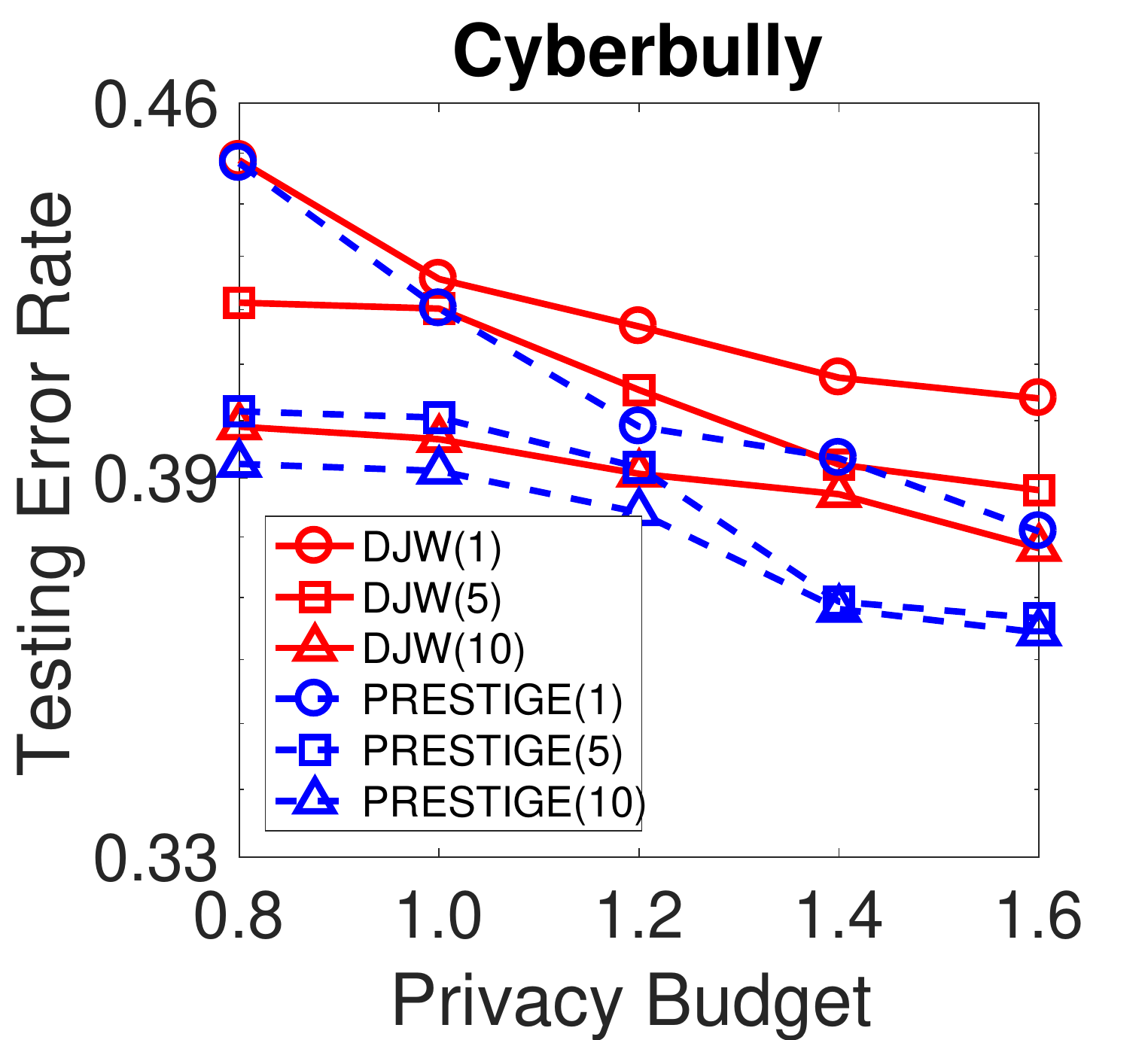} &
\includegraphics[width=0.5\textwidth]{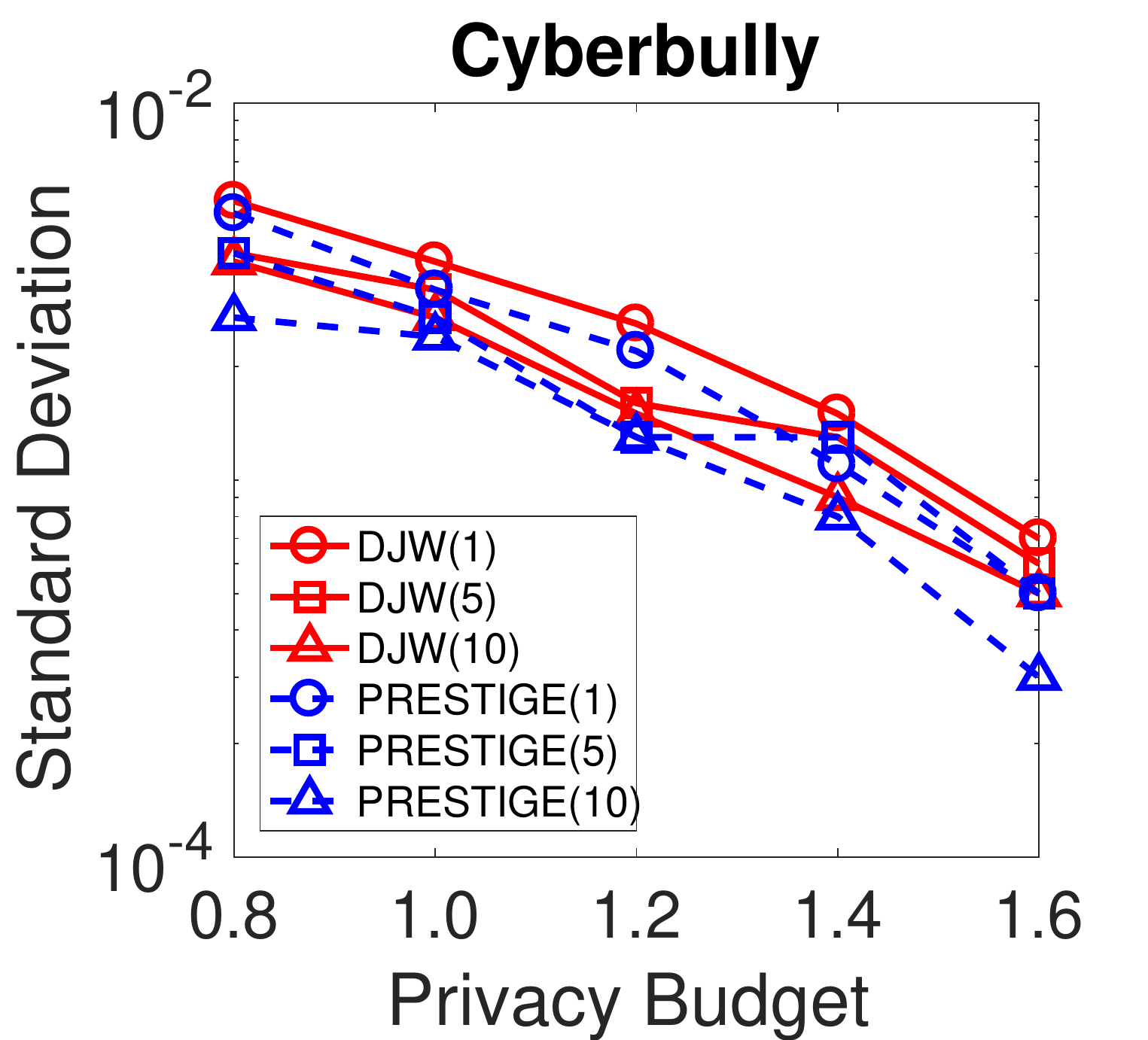} \\
\includegraphics[width=0.5\textwidth]{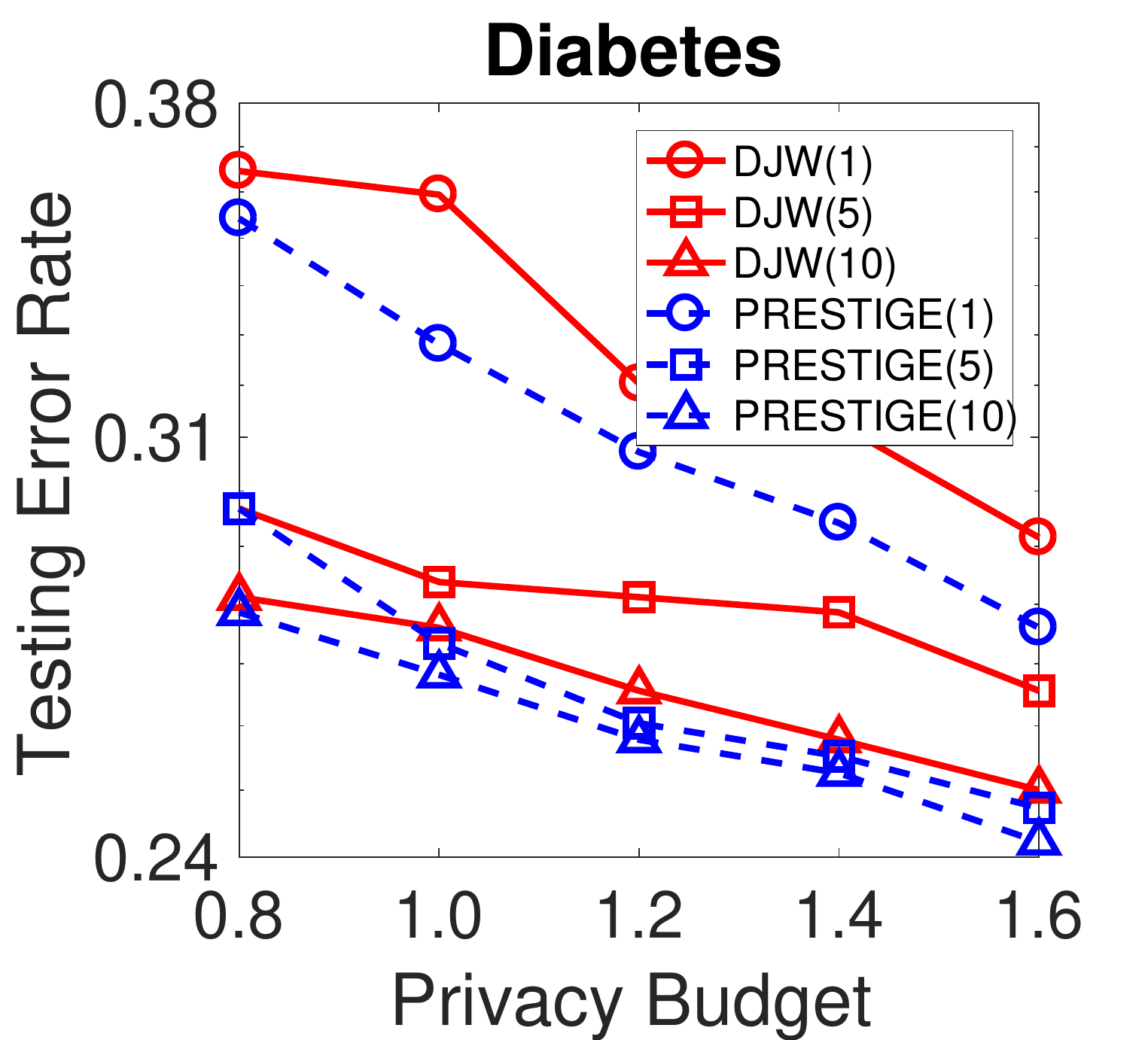} &
\includegraphics[width=0.5\textwidth]{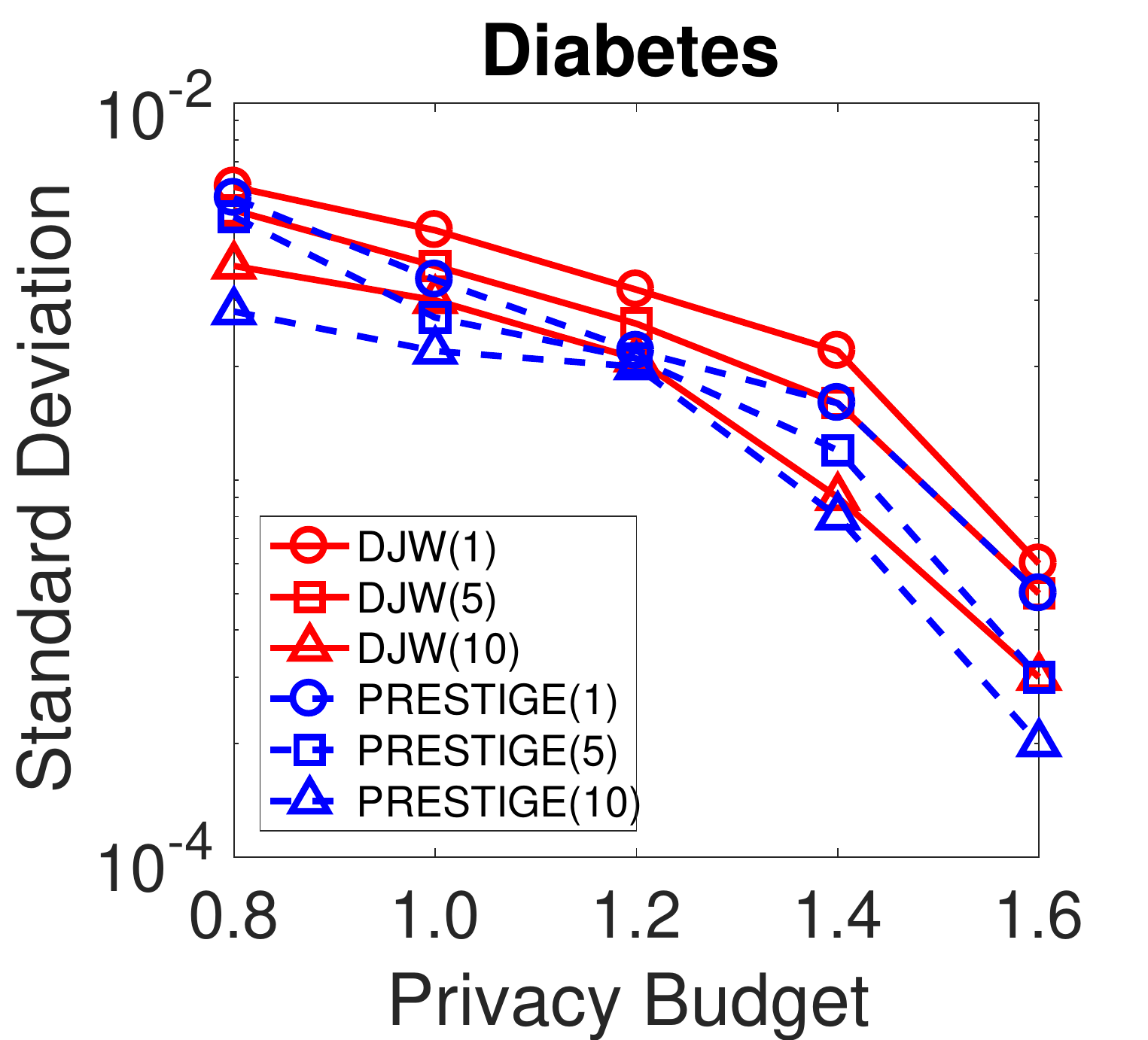}
\end{tabular}
\caption{To verify the \emph{effectiveness of mini-batch trick} in real-world situations, we compare the mini-batch version of PRESTIGE and DJW, where the batch size (listed in the parentheses) is $1$, $5$ and $10$ respectively. We provide the testing error rate and standard deviation with the privacy budget $\varepsilon$ on two real-world datasets.}
\label{Effectiveness-MiniBatch-Real}
\end{figure}

Before delving into empirical results, we first introduce two strands of experimental datasets. One strand comes from UCI datasets \cite{chang2011libsvm}, which include small-scale \textit{A7A}, middle-scale \textit{MNIST38}, large-scale \textit{SUSY} and high-dimensional \textit{REAL-SIM}. We use these datasets to verify the proposed PRESTIGE, since these datasets are representative in the scalability and dimensionality. The other strand comes from two large-scale sensitive datasets in the real world, namely social networks \cite{yuan2010personalized} and healthcare \cite{gelfand2012privacy}. As these domains are highly related to individuals in real-world situations, they are very suitable to verify the practicality of PRESTIGE.

In social networks, we use \textit{Cyberbully} detection dataset provided by Australia Research Alliance for Children \& Youth (\url{https://www.aracy.org.au/}). Cyberbully detection dataset is provided by our research collaborator from Australia Research Alliance for Chidren \& Youth (ARACY). The raw dataset consists of $120,000$ twitter posts. After being preprocessed by skip-thought vectors \footnote{\url{https://github.com/tensorflow/models/tree/master/research/skip_thoughts}}, each post is represented as a vector with $4,800$ dimension.

In healthcare, we employ \textit{Diabetes} prediction dataset extracted from MarketScan Research Databases (\url{https://truvenhealth.com/}). Note that, for imbalanced datasets, the ratio between positive and negative samples is rescaled around $1:1$. The statistics for both UCI and real-world datasets are summarized in Table~\ref{datasets}.

\subsubsection{Robustness of PRESTIGE}\label{Performance of POSTAL}
To verify the robustness of PRESTIGE, in Fig.~\ref{Robust-UCI} and \ref{Robust-Real}, we provide the performance (the testing error rate (TER) with the privacy budget) comparison among PRESTIGE, vanilla SGD and DJW on four UCI datasets and two real-world datasets. Inspired by the idea of control variable, we leverage the same hinge loss to realize PRESTIGE, vanilla SGD and DJW.

We observe that, under the same hinge loss, PRESTIGE obviously achieves the lower TERs than DJW along the entire x-axis, which demonstrates the robustness of PRESTIGE. By our analysis, the good results of PRESTIGE should be due to the mechanism of curriculum learning, which learns a reliable model from the ordered label sequence. Meanwhile, with the increase of privacy budget, the TERs of both PRESTIGE and DJW decrease continuously. According to private sampling strategy, $\frac{e^{\varepsilon}}{e^{\varepsilon} + 1}$ will approach to $1$ with the increase of privacy budget. It means that, with the large probability, $Q = 1$ and the rotation angle $\langle G_p, g\rangle$ (from gradient $g$ to private gradient $G_p$) will be restricted under $90^{\circ}$ due to line $15-16$ in Algorithm~\ref{algorithm-PRESTIGE}. Equivalently, the noise injection will decrease, which lowers the TERs of both methods. Note that, the TERs of both PRESTIGE and DJW are worse than the TER of vanilla SGD, since all privacy-preserving algorithms essentially equals to noise injection. However, vanilla SGD does not preserve any privacy, and cannot be directly used in sensitive data.

\subsubsection{Effectiveness of Mini-Batch Trick}\label{Fast-Convergence} We further verify the effectiveness of mini-batch trick on the same datasets, we compare the mini-batch version of PRESTIGE and DJW, where the batch size is $1$, $5$ and $10$ respectively. We simultaneously provide the TER and standard deviation (SD) under the privacy budget. In Fig.~\ref{Effectiveness-MiniBatch-UCI} and \ref{Effectiveness-MiniBatch-Real}, we have three-fold observation with an increase in the privacy budget. Firstly, the TERs and SDs of mini-batch PRESTIGE and mini-batch DJW decrease continuously with the increase of privacy budget. Secondly, under the same method, the larger the batch size is, the lower the TERs and SDs are. Therefore, it is effective to leverage the mini-batch trick for the better robustness. Lastly, under the same batch size, PRESTIGE acquires the lower TERs and SDs than DJW, which again validates the robustness of PRESTIGE.

There are two points to be noted that: 1) it is a common sense that the increase of batch size will decrease the convergence rate of stochastic optimization \cite{li2014efficient}. Therefore, we will explore the optimal tradeoff between convergence and robustness of PRESTIGE in the near future. 2) In Fig.~\ref{Effectiveness-MiniBatch-UCI} and \ref{Effectiveness-MiniBatch-Real}, the advantages of PRESTIGE in REAL-SIM, Cyberbully and Diabetes are not as significant as that in other datasets. The reason may be that, the dimension of these datasets are higher, and the noise injection is more adverse for training robust models. Thus, it becomes more difficult for curriculum learning to solve the issue of robust degeneration.

\subsubsection{Efficacy of Different Losses}
Lastly, we verify the efficacy of different losses on the same datasets. Table~\ref{loss-comparison} presents the testing error rate (TER) and the standard deviation (SD) under the privacy budget equal to $1$. According to the results, we derive the following conclusions. For DJW, the robustness (TER$\pm$SD) of logistic loss outperforms that of another three losses in most cases, since logistic loss is convex and smooth, and it can be optimized robustly. For PRESTIGE, the robustness of Gompertz loss has a marginal advantage over that of logistic loss. The reason may be that, private curriculum in PRESTIGE introduces some noisy labels, and Gompertz loss can reduce robustness degeneration caused by these noisy labels on each update of the primal variable \cite{han2016convergence}. It is noted that, for DJW on the high-dimensional \textit{REAL-SIM}, the robustness of Gompertz loss surpasses that of logistic loss. We believe that, for high-dimensional datasets, the robustness is easily degenerated by ``noise injection'', which comes from private sampling. Due to \cite{han2016convergence}, Gompertz loss can reduce robustness degeneration caused by such noise.

\begin{table*}[!tp]
\caption{Efficacy of different losses on both UCI and real-world datasets. In each dataset, results of DJW are in the first row, while results of PRESTIGE are in the second row.}
\centering
\scalebox{1}{
\begin{tabular}{c|c|c|c|c}
	\hline
	 Dataset &  Hinge Loss & Gompertz Loss & Logistic Loss & Ramp Loss \\ \hline
	 \textit{A7A}  & 0.1980$\pm$0.0019 & 0.1991$\pm$0.0014 & \bf{0.1892$\pm$0.0002}  &  0.2184$\pm$0.0229           \\ \cline{2-5}
	     & 0.1791$\pm$0.0009 & \bf{0.1732$\pm$0.0002}                       & 0.1738$\pm$0.0002 & 0.2126$\pm$0.0294 \\ \hline
	 \textit{MNIST38}   & 0.0953$\pm$0.0003 & 0.1066$\pm$0.0001 & \bf{0.0918$\pm$0.0002}  & 0.5005$\pm$0.0001           \\ \cline{2-5}
	    & 0.0763$\pm$0.0002 & \bf{0.0746$\pm$0.0002}                       & 0.0796$\pm$0.0002 & 0.4005$\pm$0.0001     \\ \hline
	\textit{SUSY}   & 0.3229$\pm$0.0005 & 0.4346$\pm$0.0001 & \bf{0.3199$\pm$0.0001}  & 0.4546$\pm$0.0002      \\ \cline{2-5}
        & 0.2733$\pm$0.0002 & \bf{0.2731$\pm$0.0001}                       & 0.3260$\pm$0.0055  & 0.4221$\pm$0.0007 \\ \hline
	\textit{REAL-SIM}   & 0.5202$\pm$0.0018 & \bf{0.4839$\pm$0.0014} & 0.5819$\pm$0.0007  & 0.5101$\pm$0.0019  \\ \cline{2-5}
        &   0.2148$\pm$0.0012   & \bf{0.2089$\pm$0.0015}                       & 0.2829$\pm$0.0029   & 0.2095$\pm$0.0029    \\ \hline
	\textit{Cyberbully} & 0.4257$\pm$0.0038 & 0.4322$\pm$0.0015 & \bf{0.4142$\pm$0.0001}         & 0.4508$\pm$0.0005 \\ \cline{2-5}
      &   0.4202$\pm$0.0032  & \bf{0.3934$\pm$0.0011}                       & 0.3989$\pm$0.0010     & 0.4459$\pm$0.0043     \\ \hline
    \textit{Diabetes}      &   0.3594$\pm$0.0060   & 0.3485$\pm$0.0072                       & \bf{0.3464$\pm$0.0026}    &  0.3750$\pm$0.0033     \\ \cline{2-5}
            & 0.3281$\pm$0.0034 & \bf{0.3125$\pm$0.0028}                       & 0.3203$\pm$0.0082         & 0.3690$\pm$0.0075  \\ \hline
\end{tabular}}
\label{loss-comparison}
\end{table*}

\section{Conclusions}\label{Conclusions}
This paper studies a private but robust SGD mechanism called PRESTIGE for the large-scale sensitive data, which aims to solve the issues of ``privacy leakage'' and ``robustness degeneration'' simultaneously. In future, we will further explore three following aspects. (1) How to speed up the convergence rate of PRESTIGE by Nesterov's accelerated strategy \cite{nesterov2007gradient} or iterative machine teaching \cite{liu2017iterative}. (2) To achieve the optimal robustness of PRESTIGE on each dataset, how to adaptively choose the ratio between $\varepsilon_r$ and $\varepsilon_s$. (3) How to extend PRESTIGE to deep learning (Appendix F). With these extensions, PRESTIGE can be applied to more situations.

\subsubsection*{Acknowledgments.} IWT was partially supported by ARC FT130100746, LP150100671 and DP180100106. XKX was supported in part by by MOE, Singapore under grant MOE2015-T2-2-069, and by NUS, Singapore under an SUG. LC was supported by ARC DP180100966. The Titan Xp used for this research was donated by the NVIDIA Corporation.

\end{document}